\documentclass[11pt,a4paper]{article}

\usepackage{graphicx}
\usepackage{indentfirst,csquotes}
\usepackage{amssymb,amsthm,amsmath,ascmac}
\usepackage{xcolor,etoolbox}
\usepackage{physics}
\usepackage{hyperref}
\usepackage{breakurl}
\theoremstyle{definition}
\newtheorem{theorem}{Theorem}[]
\newtheorem{definition}[theorem]{Definition}
\newtheorem{example}[theorem]{Example}
\newtheorem{lemma}[theorem]{Lemma}
\newtheorem{proposition}[theorem]{Proposition}

\newtheorem{remark}[theorem]{Remark}

\numberwithin{theorem}{section}  
\numberwithin{equation}{section} 

\hypersetup{ colorlinks=false, linkcolor=black, filecolor=black, urlcolor=black }

\allowdisplaybreaks

\title{Special Features of the Upper Half-Plane in Semiclassical Theory}

\author{Hirotaka Wakuta\thanks{mf22144@shibaura-it.ac.jp} \\ [0.5cm]
  Graduate School of Engineering and Science, \\
  Shibaura Institute of Technology, \\
  Japan.
  }
\date{}

\begin{document}

\maketitle

\begin{abstract}
  We study conditions that the semiquantised Riemannian metric $g_Q$ and Levi-Civita connection $\nabla_Q$ have the same form as classical one in semiclassical theory.
  Concrete examples of semiclassical theory have been computed by Majid et al. 
  for upper half-planes, hemispheres \cite{name1} and complex projective spaces \cite{name2}.
  However, only in the example of the upper half-planes the semiquantised metric and connection have the same form as classical one.
  In this study, we first investigate the conditions that the Riemannian metric has the same form as classical one 
  when its components are replaced by general functions.
  Next, based on these conditions we compute the generalised Ricci form and 
  investigate the conditions that a quantised connection is a quantum Levi-Civita connection.
  Finally, we semiquantise the upper half-plane with Poincar\'{e} metric generalised by a parameter $t$.
  This paper is based on my Master thesis.
\end{abstract}

\section{Introduction}
Classical mechanics has been formulated geometrically by using symplectic geometry.
On the other hand, the geometrical framework corresponding to quantum mechanics has not yet been completed.
Noncommutativity is important in quantum mechanics, and in order to construct a geometry corresponding to quantum mechanics, it is necessary to extend the conventional geometry to a noncommutative one.
One of the most actively studied is noncommutative geometry proposed by Alain Connes.
The research method is to study $C^*$-algebras, which are noncommutative algebras to study noncommutative spaces.
This is a top-down approach.
Another method of constructing noncommutative geometry is to regard quantum mechanics as a ``quantized version of classical mechanics'', and to construct noncommutative geometry by quantizing classical objects.
This is a bottom-up approach.
There are several types of quantization, but in this study, we focus our attention on deformation quantization.
Deformation quantization is a method to construct noncommutative geometry by introducing a formal power series with a deformation parameter (as small as Planck's constant) for a classical quantity, and by using a noncommutative product, the $*$-product.
However, convergence problems arise when dealing with formal power series.
Furthermore, it is currently difficult to experimentally observe quantities above the second order of Planck's constant.
In the latter part of this paper, we will construct a noncommutative Riemannian geometry based on the semiclassical theory of quantization by neglecting terms above the second order of the deformation parameters.

And in this paper, we treat the upper half-plane as a concrete example of the semiclassical theory.
In \cite{name1}, concrete examples of hemispheres and upper half-planes are computed, 
and in \cite{name2}, concrete examples of complex projective spaces are computed.
However, only for the upper half-plane, the quantized results have the same form as classical one.
In this study, we investigate the conditions for having the same form as classical one 
when the Riemannian metric is given in the form of a function that is a further generalization of the Poincar\'{e} metric.
Furthermore, based on the conditions for havcing to the classical form, we generalize the Poincar\'{e} metric with the parameter $t$ and semiquantize the Riemannian metric and the Levi-Civita connection.

\section{Classical differential geometry}   \label{sec:classicaldifferentialgeometry}
In this section we look at the basic issues of ordinary commutative differential geometry.
First, the basics of vector bundles on manifolds are discussed.
This is because in semiclassical theory we first discuss the noncommutative version corresponding to classical vector bundles.
Next, we discuss the basics of Riemannian manifolds in the context of classical differential geometry, restricting ourselves to tangent bundles, which are examples of vector bundles.
We then redefine notion of connection on Riemannian manifolds using differential forms, which is not done in classical differential geometry.
This view is necessary for quantum Riemannian geometry.
Finally, we look at definitions and examples of symplectic and Poisson manifolds, which are necessary for quantization.

\subsection{Fundamentals of vector bundles}    \label{subsec:vecbdl}
\subsubsection{Definitions and examples}
In a bimodule approach, a bimodule connection is constructed on a bimodule.
This is the connection corresponding to the connection on the vector bundle.
Thus, we will first look at the definition and examples of vector bundles.
In this section, $M$ is a $n$-dimensional smooth manifolds.
\begin{definition}
  (Vector bundle)
  Let $M$ be a manifold.
  A vector bundle over $M$ is
\begin{enumerate}
    \item a manifold $E$.
    \item a continuous map $\pi : E \to M$.
    \item for each $p \in M$, the fiber $E_p = \pi^{-1}(p)$ on $p$ is a constant-dimensional vector space.
    \item for each $p \in M$ there is a coordinate neighborhood $U$ of $p$, and there exists a diffeomorphism map
          \begin{equation*}
            \varphi : \pi^{-1}(U) \to U \times \mathbb{R}^r
          \end{equation*}
          and for any $q \in U$ the restriction of $\varphi$ to $\pi^{-1}(q)$ gives a linear homomorphism map
          \begin{equation*}
            \varphi|_{\pi^{-1}(q)} : \pi^{-1}(q) \to \{q\} \times \mathbb{R}^r.
          \end{equation*}
  \end{enumerate}
\end{definition}
Given a vector bundle $E$ and an coordinate neighbourhood $U$ of $M$, we can consider sections of $E$ on $U$.
\begin{definition}
  (Section)
  A section of $E$ is a smooth map $\xi : M \to E$ such that $\pi(\xi(p)) = p\;(p \in M)$.
  We denote $\Gamma(E)$ as a set of all sections of the vector bundle $E$.
\end{definition}

The tangent bundle is an example of a vector bundle.
\begin{example}
  $TM := \bigcup_{p \in M} T_pM$ is called the tangent bundle, 
  where $T_pM$ is the tangent vector space at $p \in M$.
  The sections of $TM$ are the vector fields.
\end{example}

We denote the ring of smooth functions by $C^{\infty}(M)$.
$\Gamma(E)$ has the following property.
\begin{lemma}   \label{lem:subsectionmodule}
  $\Gamma(E)$ is a module over $C^{\infty}(M)$.
\end{lemma}
In the commutative case, the action $f\xi$ gives $\Gamma(E)$ a module over $C^{\infty}(M)$.
However, for an action $f*\xi$ by a noncommutative product $*$, $\Gamma(E)$ is not a module.

\subsubsection{Connections on a vector bundle}   \label{subsubsec:connectionoverbdl}
In this section, we define connection on a vector bundle.
\begin{definition}
  (Connection)
  A connection on vector bundle $E$ is a linear map $\nabla : \Gamma(E) \to \Gamma(T^*M \otimes E)$ satisfying the following Leibniz rule
  \begin{equation}
    \nabla(f\xi) = df \otimes \xi + f \nabla \xi    \label{eq:connectionleibnizrule}
  \end{equation}
  for $f \in C^{\infty}(M)$ and $\xi \in \Gamma(E)$.
\end{definition}
Suppose $(x^1, \ldots, x^n)$ is a local coordinate system on $M$
and choose the natural basis of $T_pM$ as $\{\pdv{x^i}\}_{i = 1, \ldots, n}$ and the natural basis of $T^*_pM$ as $\{dx^i\}_{i = 1, \ldots, n}$.
Using a dual pairing
\begin{equation*}
  \left\langle dx^i, \pdv{x^j} \right\rangle = dx^i \qty(\pdv{x^j}) = \delta_j^i,
\end{equation*}
we can consider the dual pairing $\nabla \xi$ and the tangent vector $X$ and write
\begin{equation*}
  \langle \nabla \xi, X \rangle = \nabla_X \xi \in \Gamma(E).
\end{equation*}
$\nabla_X \xi$ is called a covariant derivative of $\xi$ along the vector field $X$ 
and $\nabla_X : \Gamma(E) \to \Gamma(E)$ is a linear map with respect to $X$ and $\xi$.
We also denote $\nabla \xi$ as a covariant derivative of $\xi$.

For a function $f$, its exterior derivative $df$ is the 1-form, 
and a dual pairing of $df$ and a vector field $X$ gives
\begin{equation*}
  \langle df, X \rangle = X(f)
\end{equation*}
from the bilinearity of the dual pairing.
Thus, taking a dual pairing of \eqref{eq:connectionleibnizrule} and a vector field $X$, we have
\begin{equation*}
  \nabla_X \xi = (X f) \xi + f \nabla_X \xi.
\end{equation*}

Vector bundle is a structure, locally a direct product of a vector space.
For a vector space, we can define a dual vector space.
Thus we consider a vector bundle that is locally a direct product of a dual vector space.
This is called a dual vector bundle and is denoted by $E^*$ for a vector bundle $E$.
A cotangent bundle is an example of a dual vector bundle.
\begin{example}
  The cotangent vector space is denoted by $T^*_pM$, which is the dual space of the tangent vector space $T_pM$ at $p \in M$.
  We call $T^*M := \bigcup_{p \in M} T^*_pM$ the cotangent bundle.
  The sections of $T^*M$ are the 1-forms.
  Moreover, the sections of $\wedge^k T^*M := \bigcup_{p \in M} (\wedge^k T^*_pM)$ are the $k$-forms.
\end{example}
Given a connection $\nabla_E$ on a vector bundle $E$, 
a connection $\nabla_E^*$ on a dual vector bundle $E^*$ is defined by
\begin{equation}
  d \langle \xi^*, \xi \rangle = \langle \nabla_{E^*} \xi^*, \xi \rangle + \langle \xi^*, \nabla_E \xi \rangle   \label{eq:dualvecbdlconnection}
\end{equation}
for $\xi \in \Gamma(E)$ and $\xi^* \in \Gamma(E^*)$.

We can also construct a tensor product bundle $E \otimes F$ over vector bundles $E, F$.
Given a connection $\nabla_E$ on a vector bundle $E$ and a connection $\nabla_F$ on a vector bundle $F$, 
a connection $\nabla_{E \otimes F}$ on the tensor product bundle is defined by
\begin{equation}
  \nabla_{E \otimes F} (\xi \otimes \xi') = (\nabla_E \otimes I_F + I_E \otimes \nabla_F)(\xi \otimes \xi') = \nabla_E \xi \otimes \xi' + \xi \otimes \nabla_F \xi'   \label{eq:tensorvecbdlconnection}
\end{equation}
for $\xi \in \Gamma(E)$ and $\xi' \in \Gamma(F)$, where $I_E, I_F$ are identity operators on the vector bundle $E$($F$), respectively.

\subsection{Riemannian manifolds}   \label{subsec:riemannstronTM}
In section \ref{subsec:vecbdl}, we discussed the connection on the vector bundle $E$.
In this section, we discuss the case $E = TM$.

A Riemannian manifold is a manifold with an inner product structure at each point on a manifold, 
and is an important object in differential geometry.
We emphasize that it is important that there exists a unique connection on a Riemannian manifold, 
called the Levi-Civita connection, which is compatible with the inner product structure.
The is because the connection compatible with the inner product structure is an important assumption 
in the semiclassical theory that will be treated later in this paper.

Hereafter in this section, we will also assume that $M$ is a smooth manifold.
The set of tangent bundle $\Gamma(TM)$ (the set of a vector field) is written as $\mathcal{X}(M)$, 
the set of cotangent bundle $\Gamma(T^*M)$ (the set of a 1-form) is written as $\Omega^1(M)$.
In addition, the set of a $k$-form is written as $\Omega^k(M)$.

\subsubsection{Affine connection}
In the discussion of section \ref{subsubsec:connectionoverbdl}, 
when the vector bundle $E$ is a tangent bundle $TM$, 
the connection on $TM$ is called an affine connection.
The definition is given once again.
\begin{definition}
  (Affine connection)
  Let $C^{\infty}(M)$ be the set of smooth function on $M$ 
  and $\mathcal{X}(M)$ be the set of vector field on $M$.
  For any $X, Y, Z \in \mathcal{X}(M)$ and $f \in C^{\infty}(M)$, 
  $\nabla_X Y$ is called a covariant derivative of the vector field $Y$ along the vector field $X$ 
  if the following four conditions are satisfied.
  \begin{enumerate}
    \item $\nabla_{Y + Z} X = \nabla_Y X + \nabla_Z X$
    \item $\nabla_{fY} X = f \nabla_Y X$
    \item $\nabla_Z (X + Y) = \nabla_Z X + \nabla_Z Y$
    \item $\nabla_Y (fX) = (Yf) X + f \nabla_Y X$
  \end{enumerate}
  Also, the map
  \begin{equation*}
    \nabla : \mathcal{X}(M) \times \mathcal{X}(M) \to \mathcal{X}(M), \quad (X, Y) \mapsto \nabla_X Y
  \end{equation*}
  is called an affine connection of $M$.
\end{definition}
Affine connections coincide with the concept of parallel transport of tangent vectors on a manifold.
In the following, affine connections (or covariant derivatives) are considered as parallel transport of tangent vectors.

\subsubsection{Curvature and torsion}   \label{subsubsec:curvatureandtorsion}
Given a parallel transport of a tangent vector, we can define the important quantities curvature and torsion.
We define curvature and torsion.
\begin{definition}
  (Curvature)
  Let $M$ be a manifold and $\nabla$ an affine connection on $M$.
  The curvature of an affine connection $\nabla$ is a tensor field defined as follows.
  \begin{equation}
    R : \mathcal{X}(M) \times \mathcal{X}(M) \times \mathcal{X}(M) \to \mathcal{X}(M), \quad (X, Y, Z) \mapsto [\nabla_X, \nabla_Y] Z - \nabla_{[X, Y]} Z   \label{eq:curvature}
  \end{equation}
\end{definition}
\begin{definition}
  (Torsion)
  Let $M$ be a manifold and $\nabla$ an affine connection on $M$.
  The torsion of an affine connection $\nabla$ is a tensor field defined as follows.
  \begin{equation}
    T : \mathcal{X}(M) \times \mathcal{X}(M) \to \mathcal{X}(M), \quad (X, Y) \mapsto \nabla_X Y - \nabla_Y X - [X, Y]    \label{eq:torsion}
  \end{equation}
\end{definition}

\subsubsection{Definition of Riemannian manifolds}    \label{subsubsec:classicalriemannianmfd}
A Riemannian manifold is defined as follows.
\begin{definition}
  (Riemannian manifold)
  A symmetric 2-tensor field $g$ on a manifold $M$ is called a Riemannian metric on $M$ 
  if $g_p : T_p M \times T_p M \to \mathbb{R}$ determines an inner product at each point $p$.
  A manifold with a given Riemannian metric is called a Riemannian manifold.
\end{definition}
There are infinitely many affine connections on manifolds.
However, there is unique a special affine connection on Riemannian manifolds called the Levi-Civita connection.
\begin{definition}
  (Levi-Civita connection)
  Let $M$ be a Riemannian manifold with Riemannian metric $g$ and let $\nabla$ be an affine connection.
  If $\nabla$ satisfies the following two conditions, we call it a Levi-Civita connection.
  \begin{enumerate}
    \item $\nabla$ is compatible with the metric $g$, i.e. $\nabla g = 0$.
    \item $\nabla$ is torsion free, i.e. $T = 0$.
  \end{enumerate}
\end{definition}
Condition 1 can be rewritten as
\begin{equation}
  X g(Y, Z) = g(\nabla_X Y, Z) + g(Y, \nabla_X Z)    \label{eq:classicalmetriccompatibilityvec}
\end{equation}
for $X, Y, Z \in \mathcal{X}(M)$.
From this equation, $\nabla$ being compatible with the metric $g$ can be interpreted as 
meaning that the parallel transport by $\nabla$ along the vector field $X$ preserves the inner product of the tangent vectors $Y, Z$.

\subsubsection{In local coordinates}
In previous sections, we discussed without using a local coordinate system.
In this section, we discuss again the contents of the previous sections using the local coordinate system.
In the following, let $M$ be an $n$-dimensional manifold and 
we denote  a coordinate system as $(x^1, \ldots, x^n)$ in the coordinate neighborhood $U$.
We also write $\partial_i := \pdv{x^i}$ for the basis of the tangent vector space $T_pM$ in $p \in M$.
In addition, we use the Einstein summation convention, which sums over the same upper and lower indices.

First, we express the covariant derivative of a vector field along a vector field using the basis $\{\partial_i\}$.
The covariant derivative of the vector field $\partial_k$ along the vector field $\partial_j$ is expressed by
\begin{equation}
  \nabla_j \partial_k := \langle \nabla \partial_k, \partial_j \rangle = \nabla_{\partial_j} \partial_k = \Gamma^i_{jk} \partial_i.    \label{eq:vecterfieldcovariantderivative}
\end{equation}
The $n^3$ function $\Gamma^i_{jk}$ on the right-hand side are called the Christoffel symbols.

Next, we compute the covariant derivative of a differential form along the vector field.
Namely, the connection $\nabla'$ on the cotangent bundle $T^*M$ is obtained using \eqref{eq:dualvecbdlconnection}.
In \eqref{eq:dualvecbdlconnection}, let $E = TM$, $E^* = T^*M$, $\xi = \partial_k$, $\xi^* = dx^i$, and dual pairing of both sides with the vector field $\partial_j$ gives
\begin{align*}
  \partial_j \langle dx^i, \partial_k \rangle & = \langle \nabla'_j dx^i, \partial_k \rangle + \langle dx^i, \nabla_j \partial_k \rangle \\
  \partial_j \delta^i_k                       & = (\nabla'_j dx^i)(\partial_k) + dx^i(\Gamma^l_{jk} \partial_l) \\
  0                                           & = (\nabla'_j dx^i)(\partial_k) + \Gamma^l_{jk} \delta^i_l \\
  (\nabla'_j dx^i)(\partial_k)                & = - \Gamma^i_{jk}.
\end{align*}
In conclusion, we have
\begin{equation}
  \nabla'_j dx^i = -\Gamma^i_{jk} dx^k.    \label{eq:formcovariantderivative}
\end{equation}
It can be seen that the covariant derivative of a differential form is slightly different from \eqref{eq:vecterfieldcovariantderivative}.
In the following, we denote $\nabla'$ as $\nabla$.

Next, we give a local representation of the Riemannian metric and the Levi-Civita connection.
First, because the Riemannian metric is a symmetric 2-tensor field, 
the Riemannian metric can be expressed as
\begin{equation}
  g = g_{ij} dx^i \otimes dx^j    \label{eq:classicalriemannianmetric}
\end{equation}
in local, where $g_{ij} := g\qty(\partial_i, \partial_j)$ is called a component of the tensor field $g$ and $g_{ij} = g_{ji}$.

Next we give the Christoffel symbols when $\nabla$ is the Levi-Civita connection.
First, the Christoffel symbols of the Levi-Civita connection is calculated by
\begin{equation}
  \Gamma^i_{jk} = \frac{1}{2} g^{il} (\partial_j g_{kl} + \partial_k g_{jl} - \partial_l g_{jk}).    \label{eq:christoffelsymbol}
\end{equation}
Also, substituting $\partial_j, \partial_k$ for $X, Y$ in \eqref{eq:torsion} respectively, the torsion is
\begin{align}
  T(\partial_j, \partial_k) & = \nabla_j \partial_k - \nabla_k \partial_j - [\partial_j, \partial_k] \nonumber \\
                            & = (\Gamma^i_{jk} - \Gamma^i_{kj}) \partial_i.    \label{eq:torsioninlocal}
\end{align}
If $\nabla$ is Levi-Civita connection, then $\Gamma^i_{jk} = \Gamma^i_{kj}$ since the Levi-Civita connection is torsion free.

It is locally written
\begin{equation}
  \partial_i g_{jk} = g_{lk} \Gamma^l_{ij} + g_{jl} \Gamma^l_{ik}   \label{eq:metriccompatibilityinlocal}
\end{equation}
that $\nabla$ is compatible.

Finally, we write the curvature using the local coordinate system and define the scalar curvature.
The Riemann curvature tensor is defined by
\begin{equation}
  R^i_{jkl} := \langle dx^i, R(\partial_j, \partial_k, \partial_l) \rangle,    \label{eq:riemanncurvaturetensor}
\end{equation}
where $\langle\cdot, \cdot\rangle$ is a dual pairing of vector field $R(\partial_j, \partial_k, \partial_l)$ and 1-form $dx^i$.
And computing the right-hand side we have
\begin{equation}
  R^i_{jkl} = \partial_k \Gamma^i_{lj} - \partial_l \Gamma^i_{kj} + \Gamma^m_{lj} \Gamma^i_{km} - \Gamma^m_{kj} \Gamma^i_{lm}.   \label{eq:riemanncurvatureinlocal}
\end{equation}
Then, using the Riemann curvature tensor, the scalar curvature is defined by
\begin{equation}
  R := g^{ij} R_{ij},    \label{eq:scalarcurvature}
\end{equation}
where $R_{ij}$ is called the Ricci curvature and is defined by
\begin{equation*}
  R_{ij} := R^k_{ikj}.
\end{equation*}
Later in this paper, we will compute the scalar curvature for the upper half-plane.

\subsection{Riemannian geometry with differential forms}
In the previous sections, we have used vector fields to define covariant derivatives, Riemannian metric, torsion, and so on.
In the context of noncommutative geometry, however, we often discuss these quantities in differential form.
Thus, the objects defined in the previous section need to be reformulated using only differential forms.
In other words, we redefine them in terms of quantities written in differential form before dual pairings with vector fields.

The covariant derivative of $dx^i$ from \eqref{eq:formcovariantderivative} is written as
\begin{equation}
  \nabla dx^i = -\Gamma^i_{hk} dx^h \otimes dx^k.    \label{eq:dx^icovariantlyderivative}
\end{equation}
In fact, a dual pairing of both sides with $\partial_j$ gives
\begin{align*}
  \langle \nabla dx^i, \partial_j \rangle & = -\Gamma^i_{hk} \langle dx^h, \partial_j \rangle dx^k \\
  \nabla_j dx^i                           & = -\Gamma^i_{jk} dx^k,
\end{align*}
which coincides with \eqref{eq:formcovariantderivative}.
Note that from the definition of the connection we take a dual pairing of $\partial_j$ and $dx^h$.

Similarly, the torsion is given by
\begin{equation}
  T = d -\wedge \nabla \quad T : \Omega^1(M) \to \Omega^2(M),    \label{eq:classicaltorsionform}
\end{equation}
from \eqref{eq:torsioninlocal}.
In fact, a dual pairing of the 2-form $T(dx^i)$ with the vector field $\partial_j, \partial_k$ gives
\begin{align*}
  ((d - \wedge \nabla) dx^i) (\partial_j, \partial_k) & = -\wedge \nabla dx^i (\partial_j, \partial_k) \\
                                                      & = -\wedge (-\Gamma^i_{mn} dx^m \otimes dx^n) (\partial_j, \partial_k) \\
                                                      & = (\Gamma^i_{mn} dx^m \wedge dx^n) (\partial_j, \partial_k) \\
                                                      & = \Gamma^i_{jk} - \Gamma^i_{kj},
\end{align*}
which coincides with \eqref{eq:torsioninlocal}.

Finally, we give the differential form-only representation of $\nabla g = 0$ (condition 1 of the Levi-Civita connection).
The differential form of $\nabla g = 0$ is written as
\begin{equation}
  \nabla g = ((\nabla \otimes id) + (\sigma \otimes id) (id \otimes \nabla))g = 0,   \label{eq:classicalmetriccompatibilityform}
\end{equation}
where $\sigma$ is called the braiding and is a linear map that swaps the order of tensor products, 
i.e. $\sigma(dx^i \otimes dx^j) = dx^j \otimes dx^i$.
If we write $g = g_{mn} dx^m \otimes dx^n$ and take a dual pairing with the vector fields $\partial_i, \partial_j, \partial_k$ in \eqref{eq:classicalmetriccompatibilityform}, 
we obtain
\begin{align*}
  & ((\nabla \otimes id) + (\sigma \otimes id) (id \otimes \nabla))(g_{mn} dx^m \otimes dx^n) (\partial_i, \partial_j, \partial_k) \\
  & = (\nabla g_{mn} dx^m) \otimes dx^n + (\sigma \otimes id) (g_{mn} dx^m \otimes \nabla dx^n) (\partial_i, \partial_j, \partial_k) \\
  & = (dg_{mn} \otimes dx^m - g_{mn} \Gamma^m_{lh} dx^l \otimes dx^h) \otimes dx^n - (\sigma \otimes id) (g_{mn} dx^m \otimes (\Gamma^n_{lh} dx^l \otimes dx^h)) (\partial_i, \partial_j, \partial_k) \\
  & = ((\partial_l g_{mn} dx^l \otimes dx^m - g_{mn} \Gamma^m_{lh} dx^l \otimes dx^h) \otimes dx^n - g_{mn} \Gamma^n_{lh} dx^l \otimes dx^m \otimes dx^h) (\partial_i, \partial_j, \partial_k) \\
  & = \partial_i g_{jk} - g_{mk} \Gamma^m_{ij} - g_{jn} \Gamma^n_{ik} \\
  & = \partial_i g_{jk} - g_{lk} \Gamma^l_{ij} - g_{jl} \Gamma^l_{ik} = 0,
\end{align*}
which coincides \eqref{eq:metriccompatibilityinlocal}.

\subsection{Symplectic manifolds}
Symplectic structures (more generally, Poisson structures) are necessary for deformation quantization.
In this section we give a definition and an example of symplectic manifolds and define Hamiltonian vector fields, 
which are necessary in semiclassical theory.

\subsubsection{Definitions and examples}
\begin{definition}
  (Symplectic manifold)
  A manifold $(M, \omega)$ is called a symplectic manifold if there exists a non-degenerate closed 2-form $\omega$ on $M$.
\end{definition}
The $\omega$ is called a symplectic form, 
and since $\omega$ is non-degenerate, the dimension of the symplectic manifold is even.

The $2n$-dimensional Euclidean space $\mathbb{R}^n$ is an important example of a symplectic manifold.
\begin{example}   \label{ex:symplecticform}
  Let $(x_1, \ldots, x_n, y_1, \ldots, y_n)$ be a coordinate system in $\mathbb{R}^{2n}$.
  And let
  \begin{equation*}
    \omega_0 = \sum_{i = 1}^n dx_i \wedge dy_i
  \end{equation*}
  be a 2-form.
  Then $(\mathbb{R}^{2n}, \omega_0)$ is a $2n$-dimensional symplectic manifold.
  $\omega_0$ is called the standard symplectic form.
\end{example}
From Darboux's theorem, any symplectic manifold is locally diffeomorphic to $(\mathbb{R}^{2n}, \omega_0)$, 
and Thus, $\mathbb{R}^{2n}$ is an important example.

\subsubsection{Hamiltonian vector fields}
On symplectic manifolds, it is possible to construct vector fields called Hamiltonian vector fields, which are related to classical mechanics.
In semiclassical theory, covariant derivative along Hamiltonian vector fields is used for semiquantization.
For this reason, we discuss Hamiltonian vector fields in this section.

Let $(M, \omega)$ be a symplectic manifold and its coordinate system $(x_1, \ldots, x_n, p_1, \ldots, p_n)$.
Here we take a smooth function $H(x, p)$ on $M$ (the function $H(x, p)$ is called a Hamiltonian).
Given a Hamiltonian, we define a Hamiltonian vector field $\hat{H}$ using the symplectic form $\omega$ as follows.
\begin{definition}    \label{def:hamiltonvectorfield}
  (Hamiltonian vector field)
  The Hamiltonian vector field $\hat{H}$ is defined by
  \begin{equation}
    i_{\hat{H}}\omega = dH   \label{eq:hamiltonvector}
  \end{equation}
  for the Hamiltonian $H$ and the symplectic form $\omega$, 
  where $i_{\hat{H}}\omega$ is the interior product of $\hat{H}$ and $\omega$.
\end{definition}
Since $\omega$ is non-degenerate,
\begin{equation*}
  \omega^{\flat} : TM \ni \hat{H} \mapsto \omega(\hat{H}, \cdot) \in T^*M
\end{equation*}
is bijective.
Therefore, there exists $\hat{H}$ such that \eqref{eq:hamiltonvector} is satisfied.

The integral curve of a Hamiltonian vector field corresponds to the trajectory of motion of a classical point mass, 
and Thus, symplectic geometry is the geometric formulation of classical mechanics.

\subsection{Poisson manifolds}  \label{subsec:poisson}
In this section, we first define Poisson algebra and show that the function ring $C^{\infty}(M)$ on symplectic manifolds has Poisson structure.
After that, we define a Poisson manifold as a generalization of symplectic manifolds.

\subsubsection{Definition of Poisson manifolds}
We first define an associative algebra before defining a Poisson algebra.
\begin{definition}
  (Associative algebra)
  Let $V$ be a vector space on a field $K$.
  If $V$ defines a product $V \times V \ni (v, w) \to vw \in V$ that satisfies the following two conditions, 
  then $V$ is said to be an algebra over a field $K$ :
  \begin{enumerate}
    \item distributive property
          \begin{equation*}
            (u + v) w = uw + vw, \quad u(v + w) = uv + uw
          \end{equation*}
    \item scalar multiplication
          \begin{equation*}
            (hu)(kv) = (hk)(uv) \quad (h, k \in K)
          \end{equation*}
  \end{enumerate}
  In particular, if the product is associative $(uv)w = u(vw)$, we call $V$ an associative algebra on a field $K$.
\end{definition}
An associative algebra is called a Poisson algebra if a bilinear map, called a Poisson bracket, is defined on it.
The definition of a Poisson algebra is as follows.
\begin{definition}
  (Poisson algebra)
  Let $A$ be an associative algebra over a field $K$.
  $A$ is a Poisson algebra if it has a bilinear product $\{\cdot, \cdot\}$ satisfying the following properties :
  \begin{enumerate}
    \item anti-symmetric
          \begin{equation*}
            \{f, g\} = -\{g, f\}
          \end{equation*}
    \item Leibniz rule
          \begin{equation*}
            \{f, gh\} = g\{f, h\} + \{f, g\}h
          \end{equation*}
    \item Jacobi identity
          \begin{equation*}
            \{f, \{g, h\}\} + \{h, \{f, g\}\} + \{g, \{h, f\}\} = 0
          \end{equation*}
  \end{enumerate}
  for any $f, g, h \in A$.
  $\{\cdot, \cdot\}$ is called the Poisson bracket of an associative algebra $A$.
\end{definition}

\subsubsection{Poisson algebra \texorpdfstring{$C^{\infty}(M)$}{}}    \label{subsubsec:poisson}
Let $(M, \omega)$ be a symplectic manifold.
Let $f, g \in C^{\infty}(M)$ be the corresponding Hamiltonian vector fields $\hat{f}, \hat{g}$ respectively.
$C^{\infty}(M)$ is an associative algebra with the following summation and multiplication.
\vspace{-5mm}
\begin{subequations}
  \begin{align*}
    \intertext{Summation}
    (f + g)(x) & = f(x) + g(x) \\
    \intertext{Multiplication}
    (fg)(x)    & = f(x) g(x)
  \end{align*}
\end{subequations}
We now define the following operator $\{ , \}$ on the associative algebra $C^{\infty}(M)$ :
\begin{equation}
  \{f, g\} = \omega(\hat{f}, \hat{g}).   \label{eq:poissonbraket}
\end{equation}
it can be shown that $\{\cdot , \cdot\}$ defined by \eqref{eq:poissonbraket} satisfies the Poisson bracket property.
Thus, $\{\cdot, \cdot\}$ is called a Poisson bracket of a symplectic manifold $M$.
\begin{proposition}
  \begin{enumerate}
    \item $\{\cdot , \cdot\} : C^{\infty}(M) \times C^{\infty}(M) \to C^{\infty}(M)$ is a bilinear map.
    \item anti-symmetric : $\{f, g\} = -\{g, f\}$
    \item Leibniz rule : $\{f, gh\} = \{f, g\}h + g\{f, h\}$
    \item Jacobi identity : $\{f, \{g, h\}\} + \{g, \{h, f\}\} + \{h, \{f, g\}\} = 0$
  \end{enumerate}
  holds for $f, g, h \in C^{\infty}(M)$.
\end{proposition}
Therefore, $C^{\infty}(M)$ is a Poisson algebra.

\subsubsection{Definitions and examples}
For symplectic manifolds, the function ring $C^{\infty}(M)$ over the symplectic manifold forms a Poisson algebra.
We define a Poisson manifold as a manifold whose function ring forms a Poisson algebra.
\begin{definition}
  (Poisson manifold)
  If the function ring $C^{\infty}(M)$ over $M$ has a bilinear map
  \begin{equation*}
    \{ \cdot, \cdot \} : C^{\infty}(M) \times C^{\infty}(M) \to C^{\infty}(M)
  \end{equation*}
  and $(C^{\infty}(M), \{ \cdot, \cdot\})$ is a Poisson algebra, then $(M, \{ \cdot, \cdot \})$ is called a Poisson manifold.
  $\{ \cdot, \cdot \}$ is called a Poisson bracket of $M$.
\end{definition}
Symplectic manifolds are even dimensional, while Poisson manifolds are also odd dimensional.
\begin{example}
  For functions $f, g$ on a $n$-dimensional Euclidean space $\mathbb{R}^n$,
  \begin{equation*}
    \{ f, g \} = \sum_{1 \leq i, j \leq n} K^{ij} x_i x_j \pdv{f}{x^i} \pdv{g}{x^j} \quad (\text{$K$ is a anti-symmetric matrix})
  \end{equation*}
  is a Poisson bracket so that $(\mathbb{R}^n, \{ \cdot, \cdot\})$ is a Poisson manifold.  
\end{example}

\subsubsection{Poisson tensor}
On symplectic manifolds, a non-degenerate closed 2-form called a symplectic form was defined.
Similarly, an anti-symmetric bivector field called a Poisson tensor is naturally derived on Poisson manifolds.
The Poisson tensor is a necessary object for semiquantization in the semiclassical theory described in the latter of this paper.
A Poisson tensor is written by the following proposition.
\begin{proposition}
  Let $\pi$ be an anti-symmetric bivector field.
  \begin{equation*}
    \{f, g\} = \pi(df, dg)
  \end{equation*}
  is a Poisson bracket if and only if $[\pi, \pi]_{SN} = 0$
\end{proposition}
$[\cdot, \cdot]_{SN}$ is called Schouten-Nijenhuis bracket.
If $M$ is a symplectic manifold, $\pi^{ij} = -\omega^{ij}$, 
where $\omega^{ij}$ is the components of the inverse matrix of symplectic form.
In this paper we deal with the upper half-plane, 
which is a symplectic manifold, 
so in the following we will use $\omega^{ij}$ for our discussion, 
but if you want to discuss it on Poisson manifolds, 
just replace it with Poisson tensor.

\section{Deformation quantization}    \label{sec:deformationquantization}
We saw in the previous section that the function ring $C^{\infty}(M)$ on a symplectic manifold $M$ is a Poisson algebra.
In this section, we deform the symplectic manifold into a quantum space by ``deforming'' the Poisson algebra $C^{\infty}(M)$.
In this section, let $A$ be an associative algebra over a field $K$.

\subsection{Deformation quantization of Poisson algebras}
We will first discuss deformations of algebras, next we will discuss deformations of Poisson algebras.

First, we define a formal power series in a general associative algebra.
\begin{definition}
  (Formal power series)
  we denote by $A[[\lambda]]$ the set of formal power series
  \begin{equation*}
    \tilde{f} = \sum_{l = 0}^{\infty} f_l \lambda^l, \quad f_l \in A
  \end{equation*}
  of $A$ with $\lambda$ as a formal parameter.
\end{definition}
$A[[\lambda]]$ is an associative algebra over $K$ by the following summation, multiplication and scalar multiplication :
\begin{subequations}
  \begin{align*}
    \intertext{summation}
    \tilde{f} + \tilde{g}     & = \sum_{l = 0}^{\infty} (f_l + g_l) \lambda^l \\
    \intertext{multiplication}
    \tilde{f} \cdot \tilde{g} & = \sum_{l = 0}^{\infty} \sum_{m + n = l} f_m \cdot g_n \lambda^l \\
    \intertext{scalar multiplication}
    \alpha \tilde{f}          & = \sum_{l = 0}^{\infty} (\alpha f_l) \lambda^l
  \end{align*}
\end{subequations}

Next, we define a deformation of an associative algebra.
\begin{definition}    \label{def:deformationofanassociativealgebra}
  (Deformation of an associative algebra)
  If a $*$-product satisfying the following four conditions is given over 
  the formal power series $A[[\lambda]]$ of $A$ over $K$, 
  then $(A[[\lambda]], *)$ is called a deformation of a algebra $A$.
  \begin{enumerate}
    \item $*$-product is a bilinear map.
    \item $*$-product is a $\lambda$ linear map, i.e.
          \begin{equation*}
            \tilde{f} * \tilde{g} = \sum_{l = 0}^{\infty} \qty(\sum_{m + n = l} f_m * g_n)\lambda^l
          \end{equation*}
          for $\tilde{f} = \sum f_l \lambda^l,\ \tilde{g} = \sum g_l \lambda^l$.
    \item For $f, g \in A$, we denote $f * g = \sum C_l(f, g)\lambda^l$. In this case, $C_0(f, g)$ satisfies
          \begin{equation*}
            C_0(f, g) = f \cdot g.
          \end{equation*}
    \item A $*$-product is associative for $f, g, h \in A$,
          \begin{equation*}
            (f * g) * h = f * (g * h).
          \end{equation*}
  \end{enumerate}
\end{definition}
\begin{remark}
  The $*$-product is also associative for $\tilde{f}, \tilde{g}, \tilde{h} \in A[[\lambda]]$,
  \begin{equation*}
    (\tilde{f} * \tilde{g}) * \tilde{h} = \tilde{f} * (\tilde{g} * \tilde{h}).
  \end{equation*}
\end{remark}
We will look at $C^{\infty}(\mathbb{R}^2)$ as an example of a deformation of an algebra.
\begin{example}   \label{ex:R2star}
  Let $B$ be a square matrix.
  We define $*_B$ as
  \begin{equation}
    f *_B g  = \sum_{l = 0}^{\infty} \frac{\lambda^l}{l!} B_{j_1} \cdots B_{j_l} f \cdot \partial_{j_1} \cdots \partial_{j_l} g,    \label{eq:starAproduct}
  \end{equation}
  then the algebra $(C^{\infty}(\mathbb{R}^2)[[\lambda]], *_B)$ is a deformation of the commutative associative algebra $C^{\infty}(\mathbb{R}^2)$.
\end{example}
If $B$ is an anti-symmetric matrix $\theta$, $*_{\theta}$ is called the Moyal product.
\begin{example}
  Let $\theta = \mqty(0 & 1 \\ -1 & 0)$.
  In this case, $*_{\theta}$ is written as
  \begin{equation*}
    f *_{\theta} g = \sum_{l = 0}^{\infty} \frac{\lambda^l}{l!} \sum_{k = 0}^l {}_lC_k (-1)^k \partial_x^{l - k} \partial_y^k f \cdot \partial_x^k \partial_y^{l - k} g
  \end{equation*}
  for $f, g \in C^{\infty}(\mathbb{R}^2)$.
\end{example}

Since the terms $C_0$ and $C_1$ are important in semiclassical theory, we discuss their properties.
First, condition 3 of the definition \ref{def:deformationofanassociativealgebra} shows that $C_0$ reflects the algebraic structure of $A$.
Next, we investigate $C_1$.
\begin{proposition}   \label{prop:Astar}
  Let $(A[[\lambda]], *)$ be a deformation of $A$.
  $C_1$ satisfies the following three properties.
  \begin{enumerate}
    \item $C_1 : A \times A \to A$ is a bilinear map.
    \item For each $f, g, h \in A$,
          \begin{equation*}
            C_1(fg, h) - f \cdot C_1(g, h) + C_1(f, g) \cdot h - C_1(f, gh) = 0.
          \end{equation*}
    \item $C_l : A \times A \to A$\;($l \geq 0$) is a bilinear map, and $C_l$ satisfies
          \begin{equation*}
            \sum_{i + j = l} C_i(C_j(f, g), h) = \sum_{i + j = l}C_j(f, C_i(g, h)).
          \end{equation*}
  \end{enumerate}
\end{proposition}
To further investigate the properties of $C_1$, we refer to $C^{\infty}(\mathbb{R}^2)$, which appeared in the example \ref{ex:R2star}.
For $f, g \in C^{\infty}(\mathbb{R}^2)$, we write $f *_A g = \sum C_l(f, g) \lambda^l$.
In this case, $C_1$ also satisfies the following property.
\begin{proposition}   \label{prop:star1order}
  For $f, g, h \in C^{\infty}(\mathbb{R}^2)$, the following holds.
  \begin{enumerate}
    \item Let $1$ be identity function, then
          \begin{equation*}
            C_0(1, f) = f, \quad C_l(1, f) = 0 \quad (l \neq 0).
          \end{equation*}
    \item If $*_{\theta}$ is Moyal product, then
          \begin{equation*}
            C_l(f, g) = (-1)^l C_l(g, f).
          \end{equation*}
    \item $C_1$ satisfies the Leibniz rule for the second argument.
          \begin{equation*}
            C_1(f, gh) = C_1(f, g)h + gC_1(f, h)
          \end{equation*}
  \end{enumerate}
\end{proposition}
In the following, we assume that the $*$-product on a deformation $(A[[\lambda]], *)$ of an algebra $A$ has the above property.
Then $C_1$ satisfies the following property for deformations of an algebra $A$.
\begin{proposition}   \label{prop:jacobi}
  Let $(A[[\lambda]], *)$ be a deformation of an algebra $A$.
  For each $f, g, h \in A$,
  \begin{equation*}
    C_1(f, C_1(g, h)) + C_1(g, C_1(h, f)) + C_1(h, C_1(f, g)) = 0
  \end{equation*}
  holds.
\end{proposition}
The following theorem follows from the proposition \ref{prop:Astar}, \ref{prop:star1order} and \ref{prop:jacobi}.
\begin{theorem}   \label{th:poisson}
  Let $A$ be a commutative associative algebra and let $(A[[\lambda]], *)$ be a deformation of $A$.
  We rewrite $C_1(\cdot , \cdot )$ in $\{\cdot , \cdot \}$, then the following holds.
  \begin{enumerate}
    \item $\{\cdot , \cdot \}$ defines a Lie algebra structure on $A$, i.e. 
          \begin{gather*}
            \{f, g\} = -\{g, f\} \\
            \{f, \{g, h\}\} + \{g, \{h, f\}\} + \{h, \{f, g\}\} = 0
          \end{gather*}
    \item Leibniz rule
          \begin{equation*}
            \{f, gh\} = \{f, g\}h + g\{f, h\}
          \end{equation*}
  \end{enumerate}
\end{theorem}
Theorem \ref{th:poisson} shows that $C_1$ has the Poisson bracket property.
Therefore, $(A, \{\cdot, \cdot\})$ is a Poisson algebra.

Next, we define a deformation of a general Poisson algebra (this is called deformation quantization).
\begin{definition}
  (Deformation quantization of Poisson algebra)
  Let $(A[[\lambda]], *)$ be a deformation of the Poisson algebra $(A, \{\cdot, \cdot\})$.
  If the $*$-product satisfies the following condition, then $(A[[\lambda]], *)$ is called a deformation quantization of the Poisson algebra $A$.
  \begin{enumerate}
    \item $* : A[[\lambda]] \times A[[\lambda]] \to A[[\lambda]]$ is a bilinear map, $\lambda$-bilinear map, and associative
          \begin{equation*}
            f * (g * h) = (f * g) * h \quad (f, g, h \in A[[\lambda]])
          \end{equation*}
          (Not necessarily commutative).
    \item For $f, g \in A$, we denote $f * g = \sum C_r (f, g) \lambda^r\;(C_r(f, g) \in A)$. Then
          \begin{equation*}
            C_0(f, g) = f \cdot g, \quad C_1(f, g) = \frac{1}{2} \{f, g\}
          \end{equation*}
          holds.
  \end{enumerate}  
\end{definition}

\subsection{Deformation quantization of symplectic manifolds}    \label{subsec:deformationquantizationofsympecticmanifolds}
We have seen that $C^{\infty}(M)$ is a Poisson algebra in section \ref{subsubsec:poisson}.
In addition, we have defined a deformation quantization of Poisson algebras in the previous section.
In this section we consider deformation quantization in the case of the Poisson algebra $C^{\infty}(M)$.

\subsubsection{Difinitions}
\begin{definition}
  (Deformation quantization of symplectic manifold)
  Deformation quantization of a symplectic manifold $M$ is a deformation quantization of the Poisson algebra \\$(C^{\infty}(M), \{ \cdot , \cdot \})$.
  Namely, the pair $(C^{\infty}(M)[[\lambda]], *)$ and the $*$-product satisfy the following conditions :
  \begin{enumerate}
    \item $* : C^{\infty}(M)[[\lambda]] \times C^{\infty}(M)[[\lambda]] \to C^{\infty}(M)[[\lambda]]$ is a bilinear map, $\lambda$-bilinear map, and associative
          \begin{equation*}
            f * (g * h) = (f * g) * h \quad (f, g, h \in C^{\infty}(M)[[\lambda]])
          \end{equation*}
          (Not necessarily commutative).
    \item For $f, g \in C^{\infty}(M)$, we denote $f * g = \sum C_r (f, g) \lambda^r$. Then
          \begin{equation*}
            C_0(f, g) = f \cdot g, \quad C_1(f, g) = \frac{1}{2} \{f, g\}
          \end{equation*}
          holds.
  \end{enumerate}
\end{definition}
In general we add an additional condition above and define it as a deformation quantization.
\begin{definition}
  (Deformation quantization)
  We call $(C^{\infty}(M)[[\lambda]], *)$ a deformation quantization of the Poisson algebra $(C^{\infty}(M), \{ \cdot , \cdot \})$ 
  if the $*$-product given by $C^{\infty}(M)[[\lambda]]$ satisfies
  \begin{enumerate}
    \item $* : C^{\infty}(M)[[\lambda]] \times C^{\infty}(M)[[\lambda]] \to C^{\infty}(M)[[\lambda]]$ is a bilinear map, $\lambda$-bilinear map, and associative
          \begin{equation*}
            f * (g * h) = (f * g) * h \quad (f, g, h \in C^{\infty}(M)[[\lambda]])
          \end{equation*}
          (Not necessarily commutative).
    \item For $f, g \in C^{\infty}(M)$, we denote $f * g = \sum C_r (f, g) \lambda^r$. Then
          \begin{equation*}
            C_0(f, g) = f \cdot g, \quad C_1(f, g) = \frac{1}{2} \{f, g\}
          \end{equation*}
          holds.
    \item Each $C_r(f, g)$ is a bidifferential operator, i.e. when the local coordinate system of $M$ is taken, $C_r(f, g)$ is denoted as
          \begin{equation*}
            C_r(f, g) = \sum_{I, J} a_{IJ}(x) \partial_I f(x) \partial_J g(x),
          \end{equation*}
          where $I, J$ are multiple indices and $\partial_I = \partial^{i_1}_1 \cdots \partial^{i_n}_n$ for $I = (i_1, \ldots, i_n)$.          
    \item $C_r(f, g) = (-1)^r C_r(g, f)$ holds for $r = 0, 1, \ldots$.
    \item $1 * f = f * 1$.
  \end{enumerate}
\end{definition}

\subsubsection{Example}
We will see the Moyal product as an example of deformation quantization of symplectic manifolds.
Let $(\mathbb{R}^{2n}, \omega_0)$ be a symplectic manifold (Example \ref{ex:symplecticform}) and 
its coordinate system be $(x_1, \ldots, x_n, y_1, \ldots, y_n)$.
From the standard symplectic form \eqref{eq:poissonbraket}, 
the following Poisson bracket is determined from $\omega_0$
\begin{equation*}
  \{f, g\}_0 = \sum_{i = 1}^n \qty(\pdv{f}{x_i} \pdv{g}{y_i} - \pdv{f}{y_i} \pdv{g}{x_i}).
\end{equation*}
This gives $(C^{\infty}(M), \{\cdot, \cdot\}_0)$ is a Poisson algebra.

Now we define the following product on the set of formal power series $C^{\infty}(\mathbb{R}^{2n})[[\lambda]]$ :
\begin{equation}
  \tilde{f} *_M \tilde{g} = \sum_{l, m = 0}^{\infty} (f_l *_M g_m) \lambda^{l + m},    \label{eq:moyalproduct}
\end{equation}
where $*_M$ on the right-hand side is defined as
\begin{align*}
  f *_M g & = f \cdot \exp \frac{\lambda}{2} (\overleftarrow{\partial}_x \overrightarrow{\partial}_y - \overleftarrow{\partial}_y \overrightarrow{\partial}_x) \cdot g \\
          & = \sum_{r = 0}^{\infty} \frac{\lambda^r}{2^r r!} f (\overleftarrow{\partial}_x \overrightarrow{\partial}_y - \overleftarrow{\partial}_y \overrightarrow{\partial}_x)^r g \\
          & = \sum_{r = 0}^{\infty} C_r(f, g)\lambda^r
\end{align*}
for $f, g \in C^{\infty}(\mathbb{R}^{2n})$.
Here $\overrightarrow{\partial}_x$ is the differential operator that 
lets $\partial_x$ act on the function to the right of it ($\overleftarrow{\partial}_x$ is similar).
The product $*_M$ defined by \eqref{eq:moyalproduct} is called the Moyal product.

\section{Semi classical thoery}    \label{sec:semiclassical}
In section \ref{sec:deformationquantization} we have discussed the general theory of deformation quantization.
In this section, we discuss the semiclassical theory introduced in \cite{name1} and \cite{name2}.
Semiclassical theory is a theory that constructs a noncommutative space using ``semiquantization'', 
which is a quantization that ignores terms of the second order or higher in the deformation parameter $\lambda$.
In other words, semiquantization works on $\mathbb{C}[[\lambda]] / \lambda^2$\ (in deformation quantization it was $\mathbb{C}$).

In this section, we first discuss a bimodule approach, which is one approach to construct noncommutative geometry.
Next, we construct tensor products, connections, and wedge products on the bimodule by semiquantizing the data on the manifold.
Finally, we construct a noncommutative Riemannian geometry by semiquantization of the Riemannian metric and others on the manifold.

Although the deformed product in \cite{name1} is denoted by $\bullet$, 
in this paper we denote it by $*$ because of its relation to the $*$-product.
Please refer to \cite{name1} for details.

\subsection{Preparation : bimodule approach}
There are several approaches to noncommutative geometry.
Typical examples are Connes's approach using the Dirac operator and cyclic cohomology, 
Van den Bergh et al.'s approach using ring-theoretic projective modules, 
quantum group approaches, but not limited to them, and within that the bimodule approach to noncommutative Riemannian geometry.
This comes from Serre-Swan theorem, 
which states that in the commutative case 
there is a one-to-one correspondence between 
(a set of sections of) a vector bundle and a (finitely generated projection) module.

The bimodule approach uses the idea described above that a module can be viewed as sections of a vector bundle.
This approach is based on the discussion of connections on vector bundles in section \ref{subsubsec:connectionoverbdl}.
In other words, we define connections on a bimodule over noncommutative algebras and 
realize Riemannian geometric structures on bimodules.
When defining a connection on a bimodule, 
the differential calculus \;$\Omega^n(A)$ is important, 
which is a bimodule consisting of forms on an associative algebra $A$.

Hereafter, $A$ is an associative algebra, and the tensor product over an algebra $A$ is denoted by $\otimes_A$.
A vector bundle is identified with a set of sections of it.
Thus, when a vector bundle $E$ is a cotangent bundle $T^*M$, 
we denote $E$ by $\Omega^1(M)$.
And we denote $\nabla_{T^*M}$ by $\nabla_{\Omega^1(M)}$.

\subsubsection{Differential forms on nonncommutative algebra}   \label{subsubsec:noncommutativeform}
As described in section \ref{sec:classicaldifferentialgeometry}, 
classical differential geometry introduces vector fields and 
uses it to define connections (covariant derivatives), metric, and so on.
Differential forms are then defined as duals of vector fields.
In the context of noncommutative geometry, however, 
it is difficult to define a noncommutative vector field.
In the commutatie case, the set of vector field, 
which is the set of sections of tangent bundle, forms a bimodule structure.
In the noncommutative case, on the other hand, 
the vector field is also a differential operator satisfying the Leibniz rule, 
so the vector field does not form a bimodule structure.
In fact, if the vector field $X$ satisfies the Leibniz rule
\begin{equation*}
  X(fg) = X(f) g + f X(g)
\end{equation*}
for $f, g \in C^{\infty}(M)$, then the action $a * X$ on $X$ by the noncommutative product $*$ of the elements $a$ of the associative algebra is
\begin{equation*}
  (a * X)(fg) = a * X(fg) = a * (X(f) g + f X(g)) = (a * X)(f) g + a * f X(g).
\end{equation*}
Due to the noncommutativity of $*$, the second term is not $f(a * X)(g)$.
In conclusion, $a * X$ does not satisfy the Leibniz rule, so $a * X$ is not a vector field.

Hence, noncommutative geometry often starts with differential forms.
This is a method of defining a differential form on a noncommutative algebra, 
and then treating the set of differential form as a bimodule.
The reason for using bimodules is that it is easy to define a connection on the bimodule.
In section \ref{subsubsec:connectionoverbdl}, a connection was a linear map satisfying the Leibniz rule \eqref{eq:connectionleibnizrule}.
However, in the noncommutative case, 
since the commutativity of $f$ and $\xi$ is not assumed, 
a connection defined for $f \xi$ is not necessarily defined for $\xi f$.
In a bimodule case, we can define a connection by defining the Leibniz rule for $f \xi$ and $\xi f$, respectively.

First, we define a 1-form on an associative algebra.
\begin{definition}
  (First order differential calculus)
  A (first order) differential calculus on an associative algebra $A$ is a pair $(\Omega^1(A), d)$ such that
  \begin{enumerate}
    \item $\Omega^1(A)$ is a $A$-bimodule.
    \item a linear map $d : A \to \Omega^1(A)$ which obeys the Leibniz rule $d(ab) = (da)b + a(db)$.
    \item $\Omega^1(A) = \operatorname{span}\{ da\;|\;a \in A \}$
  \end{enumerate}
  The elements of $\Omega^1(A)$ are called 1-forms.
\end{definition}
Next, we also define a $n$-form on an associative algebra.
\begin{definition}    \label{def:differentialcalculus}
  (Differential calculus)
  A differential calculus on an associative algebra $A$ is a pair $(\Omega^n(A), d)$ such that 
  an associative wedge product $\wedge : \Omega^n(A) \otimes_A \Omega^m(A) \to \Omega^{n + m}(A)$ and 
  the linear map $d$ satisfy the following
  \begin{enumerate}
    \item $\Omega^0(A) = A$
    \item $d^2 = 0$
    \item $d(\xi \wedge \eta) = d\xi \wedge \eta + (-1)^{\abs{\xi}} \xi \wedge d\eta$\ ($\abs{\xi} = n$ if $\xi \in \Omega^n(A)$)
    \item $\Omega^n(A)$ is generated by $A, dA$.
  \end{enumerate}
  The elements of $\Omega^n(A)$ are called $n$-forms.
\end{definition}
\begin{remark}   \label{rem:ordercommutatibity}
  In classical case, a wedge product is graded commutative, i.e. $\xi \wedge \eta = (-1)^{\abs{\xi}\abs{\eta}} \eta \wedge \xi$,
  however, we do not assume graded commutativity for a wedge product defined on $\Omega^n(A)$.
\end{remark}

\subsubsection{Connections on bimodules}   \label{subsubsec:moduleconnection}
In the bimodule approach, a noncommutative vector bundle is regarded as a bimodule whose action is defined by a noncommutative product.
In this section, we define a connection on a bimodule.
First, we define a left connection.
\begin{definition}
  (Left connection)
  Let $E$ be a $A$-bimodule.
  A left connection $\nabla_E$ on $E$ to be a map $\nabla_E : E \to \Omega^1(A) \otimes_A E$ obeying the left Leibniz rule
  \begin{equation}
    \nabla_E (ae) = da \otimes_A e + a \nabla_E (e).   \label{eq:leftconnection}
  \end{equation}
\end{definition}
The left Leibniz rule \eqref{eq:leftconnection} corresponds to the Leibniz rule \eqref{eq:connectionleibnizrule},
which a connection on a vector bundle satisfies.
Since a left connection is defined for a left $A$-module, we can also define a left connection for a $A$-bimodule.
However, in a $A$-bimodule there is also a right action $ea \in E$ of $a \in A$ on $e \in E$.
Thus, we call $\nabla_E$ a bimodule connection, if the left connection $\nabla_E$ can also define a connection for the right action $ea$.
\begin{definition}
  (Bimodule connection)
  Let $\nabla_E$ be a left connection on $A$-bimodule.
  We call $\nabla_E$ a bimodule connection if there exists a bimodule map $\sigma_E : E \otimes_A \Omega^1(A) \to \Omega^1(A) \otimes_A E$ that satisfies
  \begin{equation}
    \nabla_E (ea) = \nabla_E (e)a + \sigma_E(e \otimes_A da).    \label{eq:bimoduleconnection}
  \end{equation}
\end{definition}
$\sigma_E$ is called a generalised braiding.

We note some remarks about $\sigma_E$.
First, $\sigma_E$ is a left module map from \eqref{eq:leftconnection} and \eqref{eq:bimoduleconnection}.
This can be seen from
\begin{align*}
  \sigma_E(be \otimes_A da) & = \nabla_E(bea) - \nabla_E(be)a \\
                        & = db \otimes_A ea + b \nabla_E(ea) - db \otimes_A ea - b \nabla_E(e) a \\
                        & = b (\nabla_E(ea) - \nabla_E(e) a) \\
                        & = b \sigma_E(e \otimes_A da)
\end{align*}
for $a, b \in A$,\ $e \in E$.
However, in general $\sigma_E$ is not a right module map.
The following proposition holds.
\begin{proposition}   \label{prop:welldefinedbimodulemap}
  A generalized braiding\;$\sigma_E$ is well-defined if and only if $\sigma_E$ is a bimodule module map.
\end{proposition}
\begin{proof}
  Let $a, b \in A$,\ $e \in E$.
  $\nabla_E(eab)$ can be written in the following two ways
  \begin{align}
    \nabla_E(eab) & = \nabla_E(ea) b + \sigma_E(ea \otimes_A db),    \label{eq:welldefinedbimodulemap1} \\
    \nabla_E(eab) & = \nabla_E(e) ab + \sigma_E(e \otimes_A d(ab)).  \label{eq:welldefinedbimodulemap2}
  \end{align}
  And then we have
  \begin{equation*}
    \sigma_E(e \otimes_A (da)b) - \sigma_E(e \otimes_A da)b = \sigma_E(ea \otimes_A db) - \sigma_E(e \otimes_A adb).
  \end{equation*}
  If $\sigma_E$ is well-defined, then $\sigma_E$ is a bimodule map because the right-hand side is zero.
  Conversely, if $\sigma_E$ is a bimodule map, then $\sigma_E$ is well-defined because the left-hand side is zero.
\end{proof}
Next, we see that $\sigma_E$ is uniquely determined if $\sigma_E$ is well-defined (i.e. a bimodule map).
\begin{proposition}
  If a generalized braiding\;$\sigma_E$ is well-defined, then $\sigma_E$ is uniquely determined.
\end{proposition}
\begin{proof}
  In the proof of Proposition \ref{prop:welldefinedbimodulemap}, 
  we rewrite \eqref{eq:welldefinedbimodulemap2} as
  \begin{equation}
    \nabla_E(eab) = \nabla_E(e) ab + \sigma'_E(e \otimes_A d(ab)).    \label{eq:welldefinedbimodulemap3}
  \end{equation}
  By assumption, we obtain
  \begin{equation*}
    \sigma'_E(e \otimes_A d(ab)) = \sigma_E(e \otimes_A d(ab))
  \end{equation*}
  from \eqref{eq:welldefinedbimodulemap1} and \eqref{eq:welldefinedbimodulemap3}.
  Thus, $\sigma_E$ is uniquely determined if it is well-defined.
\end{proof}
If $\sigma_E$ exists, from \eqref{eq:leftconnection} and \eqref{eq:bimoduleconnection} one can also deduce a useful formula
\begin{equation}
  \sigma_E(e \otimes_A da) = da \otimes_A e + \nabla_E [e, a] + [a, \nabla e].   \label{eq:braidingformula}
\end{equation}
If the action of $A$ on a $A$-bimodule $E$ is commutative, 
the second and third terms of \eqref{eq:braidingformula} are zero.
In this case, the generalized braiding\;$\sigma_E$ is simply a map that swaps the order of tensor products.
This is a generalization of \eqref{eq:classicalmetriccompatibilityform}.

\begin{remark}
  Above, we defined from a left connection to define a bimodule connection.
  However, there is another style of connection called a right connection $\nabla_E : E \to E \otimes_A \Omega^1(A)$, 
  which is defined by
  \begin{equation*}
    \nabla_E(ea) = e \otimes_A da + \nabla_E (e) a
  \end{equation*}
  on a right $A$-module, 
  and is called a bimodule connection if there exists a bimodule map satisfying
  \begin{equation*}
    \nabla_E(ae) = a \nabla_E (e) + \sigma_E(da \otimes_A e).
  \end{equation*}

  This is the difference between defining a connection on a vector bundle in classical differential geometry 
  as $\nabla : \Gamma(E) \to \Gamma(T^*M \otimes E)$ or
  $\nabla : \Gamma(E) \to \Gamma(E \otimes T^*M)$.
\end{remark}

The bimodule connection on the tensor product $E \otimes_A F$ of the bimodule $E, F$ is defined by
\begin{equation}
  \nabla_{E \otimes_A F} (e \otimes_A f) = \nabla_E e \otimes_A f + (\sigma_E \otimes_A id) (e \otimes_A \nabla_F f).    \label{eq:tensorbimoduleconnection}
\end{equation}
\eqref{eq:tensorbimoduleconnection} corresponds to \eqref{eq:tensorvecbdlconnection} in the case of vector bundles.

In order to semiquantitize a torsion later, 
we define a torsion using a connection on a module.
A torsion on a bimodule is defined as
\begin{equation}
  T_{\nabla} = d - \wedge \nabla \quad T_{\nabla} : \Omega^1(A) \to \Omega^2(A)   \label{eq:classicaltorsion}
\end{equation}
with reference to \eqref{eq:classicaltorsionform}.

\subsubsection{Riemannian structure over bimodules}    \label{subsubsec:riemmannianstronbimodule}
In this section, as in section \ref{subsec:riemannstronTM}, 
the bimodule $E$ is $\Omega^1(A)$ and 
we see a Riemannian metric and the Levi-Civita connection.
A Riemannian metric on a bimodule $\Omega^1(A)$ is called a quantum Riemannian metric and is defined as follows.
\begin{definition}    \label{def:quantumriemannianmetric}
  (Quantum Riemannian metric)
  Let $A$ be an associative algebra. \\
  $g \in \Omega^1(A) \otimes_A \Omega^1(A)$ is called a quantum Riemannian metric if the following two conditions are satisfied.
  \begin{enumerate}
    \item $g$ is invertible, in the sense that there exists a bimodule map $(\cdot, \cdot) : \Omega^1(A) \otimes_A \Omega^1(A) \to A$ such that
          \begin{equation}
            (\omega, g^{(1)})g^{(2)} = \omega = g^{(1)} (g^{(2)}, \omega)   \label{eq:quantumriemanniancondition}
          \end{equation}
          for all $\omega \in \Omega^1(A)$, where we write $g = g^{(1)} \otimes_A g^{(2)}$.
    \item $g$ is symmetric, i.e. $\wedge(g) = 0$.
  \end{enumerate}
\end{definition}
Condition 1 states that the contractions from left and right are equal by the quantum Riemannian metric $g$.
And condition 1 is equivalent to the centrality of the metric.
\begin{proposition}   \label{prop:centralandinvertible}
  A quantum Riemannian metric $g$ is central in a bimodule if and only if there exists a bimodule map $(\cdot, \cdot)$ which satisfies \eqref{eq:quantumriemanniancondition}.
\end{proposition}
\begin{proof}
  We denote $g = g^{(1)} \otimes_A g^{(2)} = g^{(1')} \otimes_A g^{(2')}$.
  Let $g$ be central in a bimodule.
  Then we show that $(a * g^{(1')}, g^{(1')}) g^{(2')} = g^{(1')} (g^{(2')} * a, g^{(1)})$.

  If $(a * g^{(1')}, g^{(1')}) g^{(2')} = a * g^{(1')}$, then
  \begin{align*}
    a * g & = a * g^{(1')} \otimes_A g^{(2')} \\
          & = (a * g^{(1')}, g^{(1)}) g^{(2)} \otimes_A g^{(2')}.
  \end{align*}
  Whereas
  \begin{align*}
    g * a & = g^{(1)} \otimes_A g^{(2)} * a \\
          & = g^{(1)} \otimes_A (g^{(2)} * a, g^{(1')}) g^{(2')} \\
          & = g^{(1)} (g^{(2)} * a, g^{(1')}) \otimes_A g^{(2')}.
  \end{align*}
  Therefore, we have
  \begin{equation*}
    (a * g^{(1')}, g^{(1')}) g^{(2')} = g^{(1')} (g^{(2')} * a, g^{(1)}).
  \end{equation*}
  Similarly, it can be shown that, given $g^{(1')} (g^{(2')} * a, g^{(1)}) = g^{(2')} * a$.
  Also, $(\cdot, \cdot)$ is a bimodule map.
  Thus, we can construct a bimodule map $(\cdot, \cdot)$ which satisfies \eqref{eq:quantumriemanniancondition}.

  Conversely, suppose there exists a bimodule map $(\cdot, \cdot)$ which satisfies \eqref{eq:quantumriemanniancondition}.
  Then we have
  \begin{align*}
    a * g & = a * g^{(1)} \otimes_A g^{(2)} \\
            & = g^{(1')} (g^{(2')}, a * g^{(1)}) \otimes_A g^{(2)} \\
            & = g^{(1')} \otimes_A (g^{(2')}, a * g^{(1)}) g^{(2)} \\
            & = g^{(1')} \otimes_A (g^{(2')} * a, g^{(1)}) g^{(2)} \\
            & = g^{(1')} \otimes_A g^{(2')} * a \\
            & = g * a.
  \end{align*}
  Thus, $g$ is central.
\end{proof}
Proposition \ref{prop:centralandinvertible} shows that in a classical case, invertibility is obvious from commutativity.
And condition 2 corresponds to the Riemannian metric $g$ being a symmetric tensor field in a classical case.

Next, we define a Levi-Civita connection on the bimodule $\Omega^1(A)$.
The Levi-Civita connection on $\Omega^1(A)$ is called a quantum Levi-Civita connection and is defined as follows.
\begin{definition}
  (Quantum Levi-Civita connection)
  Let $A$ be an associative algebra and $\nabla$ be a bimodule connection on $\Omega^1(A)$.
  Moreover, let $g \in \Omega^1(A) \otimes_A \Omega^1(A)$ be a quantum metric.
  $\nabla$ is called a quantum Levi-Civita connection if $\nabla$ satisfies the following two conditions.
  \begin{enumerate}
    \item $\nabla$ is compatible with $g$, i.e. $\nabla g = 0$.
    \item $\nabla$ is torsion free.
  \end{enumerate}
\end{definition}
It can be seen that the definition of a quantum Levi-Civita connection is the same 
as the definition of a Levi-Civita connection described in section \ref{subsubsec:classicalriemannianmfd}.
Here, the torsion in condition 2 is defined by \eqref{eq:classicaltorsion}.

Finally, we define a weak quantum Levi-Civita connection, 
which is a generalization of a quantum Levi-Civita connection \cite{name4}.
We define a cotorsion before defining a weak quantum Levi-Civita connection.
\begin{definition}
  (Cotorsion)
  The cotorsion of a connection $\nabla$ with quantum metric $g$ is the element $\mathrm{co}T_{\nabla} \in \Omega^2(A) \otimes_A \Omega^1(A)$ defined by
  \begin{equation*}
    \mathrm{co}T_{\nabla} = (d \otimes id - (\wedge \otimes id) \circ (id \otimes \nabla))g.
  \end{equation*}

  A connection is called cotorsion free if $\mathrm{co}T_{\nabla} = 0$.
  If $\nabla$ is torsion free, $\mathrm{co}T_{\nabla}$ is written as
  \begin{equation}
    \mathrm{co}T_{\nabla} = (\wedge \otimes id) \nabla g.   \label{eq:cotorsion_2}
  \end{equation}
\end{definition}
We define a weak quantum Levi-Civita connection.
\begin{definition}
  (Weak quantum Levi-Civita connection)
  A weak quantum Levi-Civita connection is a connection $\nabla$ which is torsion free and cotorsion free.
\end{definition}

We describe the origin of a cotorsion and 
see that a weak quantum Levi-Civita connection is a generalization of a quantum Levi-Civita connection.
In affine geometry, an affine connection $\nabla^*$ called a dual connection for an affine connection $\nabla$ on a manifold $M$ is defined by
\begin{equation*}
  X g(Y, Z) = g(\nabla_X Y, Z) + g(Y, \nabla^*_X Z) \quad (X, Y, Z \in \Gamma(TM)).
\end{equation*}
The torsion $T^*$ of $\nabla^*$ is computed by
\begin{equation}
  (\nabla_Y g) (X, Z) + g(T^*(X, Y), Z) = (\nabla_X g) (Y, Z) + g(T(X, Y), Z),   \label{eq:cotorsion}
\end{equation}
which corresponds to $\mathrm{co}T_{\nabla}$.
Also, if $\nabla$ and $\nabla^*$ is torsion free in \eqref{eq:cotorsion}, then
\begin{equation*}
  (\nabla_Y g) (X, Z) = (\nabla_X g) (Y, Z)
\end{equation*}
holds.
This equation is called Codazzi equation, 
and $\nabla$ satisfying this is written as $\nabla g = C$, where $C$ is a totally symmetric tensor field.
When $C = 0$, $\nabla$ coincides with $\nabla g = 0$, which is compatible with the Riemannian metric.
Namely, if $\nabla$ and $\nabla^*$ are torsion free, 
then we obtain a connection satisfying $\nabla g = C$, which weakens the condition $\nabla g = 0$.

A weak quantum Levi-Civita connection is a generalized quantum Levi-Civita connection in this sense.

\subsection{Deformation of \texorpdfstring{$C^{\infty}(M)$}{}}
Starting from this section, we will construct our bimodule approach by semiquantizing classical data.
The first step is to prepare a noncommutative algebra.

In the following sections, we will assume that $M$ is a symplectic manifold for simplicity.
This is because we can assume that $M$ is a Poisson manifold (the Poisson tensor $\pi^{ij}$ is degenerate), 
however, in this paper, we treat the upper half-plane, a symplectic manifold, as a concrete example.

If $M$ is a symplectic manifold, a deformation of the algebra $C^{\infty}(M)$ 
is similar to that in section \ref{subsec:deformationquantizationofsympecticmanifolds}.
The Poisson algebra $(C^{\infty}(M), \{\cdot, \cdot\})$ is deformed to $(C^{\infty}(M)[[\lambda]], *)$ by the following $*$-product
\begin{equation*}
  a * b = ab + \frac{\lambda}{2} \{ a, b \} \quad (a, b \in C^{\infty}(M)),
\end{equation*}
where $\{\cdot, \cdot\}$ is a Poisson bracket.
The $*$-product is associative at $\mathcal{O}(\lambda^2)$.
\begin{align*}
  (a * b) * c & = \qty(ab + \frac{\lambda}{2} \{a, b\}) * c \\
              & = (ab)c + \frac{\lambda}{2} \{ab, c\} + \frac{\lambda}{2} \{a, b\}c \\
              & = a(bc) + \frac{\lambda}{2} (a\{b, c\} + \{a, c\}b) + \frac{\lambda}{2} (\{a, bc\} - b\{a, c\}) \\
              & = a(bc) + \frac{\lambda}{2} \{a, bc\} + \frac{\lambda}{2} a\{b, c\} \\
              & = a * \qty(bc + \frac{\lambda}{2} \{b, c\}) \\
              & = a * (b * c)
\end{align*}

Then the commutator for the $*$-product of $a$ and $b$ is
\begin{equation*}
  [a, b]_* = \lambda \{a, b\}.
\end{equation*}
In the following, we denote $C^{\infty}(M)[[\lambda]] = A_{\lambda}$ and assume that the associative algebra is $A_{\lambda}$.
Note that the following facts are working at first order in $\lambda$ and dropping errors $\mathcal{O}(\lambda^2)$.

\subsection{Deformation of bimodules}    \label{subsec:functionandform}
In the previous section we defined noncommutative algebra $A_{\lambda}$.
Next, we define a $A_{\lambda}$-bimodule.

It was mentioned before that a vector bundle can be regarded as a (projective) module.
We construct a noncommutative vector bundle by noncommutative action of a bimodule.
In other words, the product between a function and a differential form is deformed to a noncommutative product.

First, the action of $C^{\infty}(M)$ on $\Omega^1(M)$ is deformed, 
and then we define the action of $A_{\lambda}$ on $\Omega^1(A_{\lambda})$ by extending it to the first order of $\lambda$.

Using the same philosophy as in section \ref{sec:deformationquantization}, we define the map $\gamma$ for $a \in C^{\infty}(M)$ and $\xi \in \Omega^1(M)$ as follows :
\begin{equation}
  [a, \xi]_* = \lambda \gamma(a, \xi).    \label{eq:gamma}
\end{equation}
From the fact that $\Omega^1(M)$ is a bimodule and associativity of the $*$-product, we have the following two formulas.
\begin{itemize}
  \item $\gamma(ab, \xi) = a \gamma(b, \xi) + \gamma(a, \xi)b$
  \item $\gamma(a, \xi b) = \gamma(a, \xi)b + \xi \{ a, b \}$
\end{itemize}
Here we define $\gamma(a, \xi) = \nabla_{\hat{a}} \xi$ ($\hat{a} = \{ a, \cdot \}$ is a Hamiltonian vector field).
In this case, the first equation can be rewritten as
\begin{equation*}
  \nabla_{\widehat{ab}} \xi = a \nabla_{\hat{b}} \xi + b \nabla_{\hat{a}} \xi,
\end{equation*}
and the left side is
\begin{equation*}
  \nabla_{\hat{a}b + a\hat{b}} \xi = b \nabla_{\hat{a}} \xi + a \nabla_{\hat{b}} \xi
\end{equation*}
from the Poisson bracket property $\{ab, c\} = a\{b, c\} + b\{a, c\}$.
Thus, the first equation is equivalent to
\begin{equation*}
  \nabla_{\hat{a}b + a\hat{b}} \xi = b \nabla_{\hat{a}} \xi + a \nabla_{\hat{b}} \xi.
\end{equation*}
This shows that $\nabla$ has tensoriality with respect to vector field direction of differentiation.
Whereas the second equation is rewritten as
\begin{equation*}
  \nabla_{\hat{a}} (\xi b) = \nabla_{\hat{a}} (\xi) b + \xi \{a, b\},
\end{equation*}
which shows that $\nabla$ satisfies the Leibniz rule.

Furthermore, we can see
\begin{itemize}
  \item $d\{ a, b \} = \gamma(a, db) - \gamma(b, da)$
\end{itemize}
because the Leibniz rule for exterior differential operator $d$ defined by $\Omega^1(M)$.
If $\nabla$ satisfies this equation, then $\nabla$ is said to satisfy Poisson compatibility.

From the above, the commutator of a function $a \in C^{\infty}(M)$ and a 1-form $\xi \in \Omega^1(M)$ is
\begin{equation}
  [a, \xi]_* = \lambda \nabla_{\hat{a}} \xi,   \label{eq:noncommutatibity}
\end{equation}
which we can realise by defining the deformed product of a function $a$ and a 1-form $\xi$ as
\begin{align*}
  a * \xi & = a\xi + \frac{\lambda}{2} \nabla_{\hat{a}} \xi, \\
  \xi * a & = a\xi - \frac{\lambda}{2} \nabla_{\hat{a}} \xi.
\end{align*}

We define $\Omega^1(A_{\lambda})$ as built on the vector space $\Omega^1(M)$ 
extended over $\lambda$ and taken with these $*$ actions $\Omega^1(A_{\lambda}) \otimes_{A_{\lambda}} A_{\lambda} \to \Omega^1(A_{\lambda})$ and 
$A_{\lambda} \otimes_{A_{\lambda}} \Omega^1(A_{\lambda}) \to \Omega^1(A_{\lambda})$ forming an $A_{\lambda}$-bimodule over $\mathbb{C}[\lambda] / \lambda^2$ or a bimodule up to $\mathcal{O}(\lambda^2)$.

In the case of a bundle $E$ with connection $\nabla_E$, the action is defined by
\begin{align*}
  a * e & = ae + \frac{\lambda}{2} \nabla_{\hat{a}} e, \\
  e * a & = ae - \frac{\lambda}{2} \nabla_{\hat{a}} e.
\end{align*}
for $e \in E$.

The noncommutativity between functions and differential forms depends on the connection $\nabla$ that appears in \eqref{eq:noncommutatibity}.
Thus, in semiclassical theory, the classical data required for semiquantization is the connection $\nabla$, 
where $\nabla$ is a covariant derivative along the Hamiltonian vector field and compatible with Poisson structure.

Finally, we describe the connection in \eqref{eq:noncommutatibity} using the local coordinate system $(x^1, \ldots, x^n)$, where the manifold $M$ is a symplectic manifold.
If $M$ is a symplectic manifold, then from the non-degeneracy of $\omega$ the symplectic form has inverse $\omega^{ij}$.
From definition \ref{def:hamiltonvectorfield}, a Hamiltonian vector field $\hat{y}$ for a Hamiltonian $y$ is given by
\begin{equation*}
  \hat{y} = \omega^{ij} \pdv{y}{x^i} \pdv{x^j}.
\end{equation*}
Thus, the covariant derivative along the Hamiltonian vector field is
\begin{equation*}
  \nabla_{\hat{y}} \xi = \omega^{ij} \pdv{y}{x^i} \nabla_j \xi
\end{equation*}
from the tensoriality of $\nabla$.

From the above, a deformed product of $a \in C^{\infty}(M)$ and $\xi \in \Omega^1(M)$ can be written
\begin{align*}
  a * \xi & = a\xi + \frac{\lambda}{2} \omega^{ij} \pdv{a}{x^i} \nabla_j \xi, \\
  \xi * a & = a\xi - \frac{\lambda}{2} \omega^{ij} \pdv{a}{x^i} \nabla_j \xi.
\end{align*}
And the commutation relation between $a$ and $\xi$ is
\begin{equation}
  [a, \xi]_* = \lambda \omega^{ij} \pdv{a}{x^i} \nabla_j \xi.    \label{eq:commutationrelation}
\end{equation}
In the following, we use the notation $a_{,i} := \pdv{a}{x^i}$ and also denote
\begin{equation*}
  [a, \xi]_* = \lambda \omega^{ij} a_{,i} \nabla_j \xi.
\end{equation*}

Also Poisson compatibility
\begin{equation*}
  d\{ a, b \} = \gamma(a, db) - \gamma(b, da)
\end{equation*}
is locally written in
\begin{equation}
  \pdv{\omega^{ij}}{x^n} + \omega^{iq} \Gamma^j_{qn} + \omega^{qj} \Gamma^i_{qn} = 0.    \label{eq:poissoncompatibility}
\end{equation}
Using the torsion tensor $T^j_{qn} = \Gamma^j_{qn} - \Gamma^j_{nq}$, we can rewrite it as
\begin{equation*}
  \nabla_n \omega^{ij} + \omega^{iq} T^j_{qn} + \omega^{qj} T^i_{qn} = 0.
\end{equation*}
If $\nabla$ is torsion free, the Poisson compatibility reduces to $\omega$ covariantly constant.
This fact will appear later when semiquantizing a connection.

Hereafter, we call the connection $\nabla$ in \eqref{eq:noncommutatibity} a background connection on $\Omega^1(M)$, 
and from the above we say that the background connection is Poisson compatibility.

\subsection{Semiquantization of tensor products}
In the previous section, we semiquantized a vector bundle by deforming a bimodule.
Next, we want to define a connection and metric on a semiquantized vector bundle, i.e. a deformed bimodule.
In order to use tensor products in connection and metric, we first semiquantitize the tensor products in this section.

Semiquantization corresponds to the construction of a functor map from a vector bundle category to a (deformed) bimodule category in catgory theory.
In other words, semiquantization is the transfer of tensor products, connections, etc. defined on a manifold map to a deformed bimodule by a functor.

In the following, we quantize the tensor product $\otimes_0$ on a manifold and construct the tensor product $\otimes_1$ on a deformed bimodule.
We will also say that ``semiquantization'' is simply ``quantization''.

Let $(E, \nabla_E)$ and $(F, \nabla_F)$ be vector bundles and denote tensor product of $E$ and $F$ by $E \otimes_0 F$.
We denote by $Q$ the functor from the category of a vector bundle to the category of $A_{\lambda}$-bimodule.
Namely, $Q(E)$ is a $A_{\lambda}$-bimodule.
We suppress writing $Q$ since it is essentially the identity on objects.
Furthermore let $Q_{E, F}$ be $q_{E, F} : Q(E) \otimes_1 Q(F) \to Q(E \otimes_0 F)$.
We assume that $e * a \otimes_1 f = e \otimes_1 a * f$ for $e \in E$ and $f \in F$, 
and quantize the tensor product $\otimes_0$ of a vector bundle by :
\begin{equation}
  q_{E, F} (e \otimes_1 f) = e \otimes_0 f + \frac{\lambda}{2} \omega^{ij} \nabla_{Ei} e \otimes_0 \nabla_{Fj} f.   \label{eq:quantizationtensorproduct}
\end{equation}
The second term is a correction term for $q_{E, F}$ to be a bimodule map.
Also, $q_{E, F}$ in \eqref{eq:quantizationtensorproduct} satisfies
\begin{equation*}
  q_{E, F}(e * a \otimes_1 f) = q_{E, F}(e \otimes_1 a * f).
\end{equation*}
This shows that $q_{E, F}$ is well-defined.
\begin{remark}
  The $q_{E, F}$ is called a natural transformation in category theory.
\end{remark}

\subsection{Semiquantization of connection}    \label{subsec:connectionquantization}
In this section, We define a bimodule connection over a $A_{\lambda}$-bimodule.

The connection $\nabla_E$ on the vector bundle $E$ is quantized using the quantized tensor product $\otimes_1$ as follows :
\begin{equation}
  \nabla_{Q(E)} = q_{\Omega^1, E}^{-1} \nabla_E - \frac{\lambda}{2} \omega^{ij} dx^k \otimes_1 [\nabla_{Ek}, \nabla_{Ej}] \nabla_{Ei},   \label{eq:nabla_QE}
\end{equation}
where $Q(\nabla_E) := \nabla_{Q(E)}$, $\Omega^1 := \Omega^1(M)$.
The second term is a correction term for $\nabla_{Q(E)}$ to satisfy the Leibniz rule \eqref{eq:leftconnection} for a $*$-product.
Therefore, $\nabla_{Q(E)}$ is a left connection.

In addition, we construct a bimodule connection by defining generalized braiding.
A braiding $\sigma_E : E \otimes \Omega^1 \to \Omega^1 \otimes E$ is semi-quantized by
\begin{equation}
  \sigma_{Q(E)}(e \otimes_1 \xi) = \xi \otimes_1 e + \lambda \omega^{ij} \nabla_j \xi \otimes_1 \nabla_{E i}e + \lambda \omega^{ij} \xi_j dx^k \otimes_1 [\nabla_{E k}, \nabla_{E i}]e    \label{eq:braiding}
\end{equation}
for $e \in Q(E)$ and $\xi \in Q(\Omega^1) = \Omega^1(A_{\lambda})$, 
where $Q(\sigma_E) := \sigma_{Q(E)}$, $\xi = \xi_j dx^j$.
\eqref{eq:braiding} is obtained by using \eqref{eq:braidingformula}.

\begin{remark}    \label{rem:braiding}
  A braiding\;$\sigma_{Q(E)}$ is a bimodule map over $A_{\lambda}$, i.e.
  \begin{equation*}
    \sigma_{Q(E)} (a * e \otimes_1 \xi * b) = a * \sigma_{Q(E)} (e \otimes_1 \xi) * b
  \end{equation*}
  for $a, b \in A_{\lambda}$, $e \in Q(E)$, $\xi \in \Omega^1(A_{\lambda})$.
  However, $\sigma_{Q(E)}$ is not a bimodule map for the usual commutative product, i.e.
  \begin{equation*}
    \sigma_{Q(E)}(ae \otimes_1 \xi b) \neq a \sigma_{Q(E)}(e \otimes_1 \xi) b.
  \end{equation*}
  We will see this in concrete calculations in the example of the upper half-plane below.
\end{remark}

From the above, we can define a bimodule connection $\nabla_{Q(E)}$ on $A_{\lambda}$-bimodule $Q(E)$.

\subsection{Semiquantization of wedge product}    \label{subsec:wedgeproduct}
According to definition \ref{def:differentialcalculus}, the wedge product is defined over the differential calculus $\Omega^n(A)$.
In this section, we define a wedge product on differential calculus $\Omega^n(A_{\lambda})$ on the associative algebra $A_{\lambda}$ using semiquantization.

As with tensor products and connections, we try to quantize a wedge product using the functor $Q$ as follows :
\begin{equation*}
  \wedge_Q : Q(\Omega^n(M)) \otimes_1 Q(\Omega^m(M)) \to Q(\Omega^n(M) \otimes_0 \Omega^m(M)) \to Q(\Omega^{n + m}(M))
\end{equation*}
\begin{equation*}
  \xi \wedge_Q \eta = \xi \wedge \eta + \frac{\lambda}{2} \omega^{ij} \nabla_i \xi \wedge \nabla_j \eta.
\end{equation*}
In this case $\wedge_Q$ is associative, but does not satisfy condition 3 in the definition of differential calculus.
\begin{proposition}
  The exterior derivative $d$ does not satisfy the Leibniz rule for $\wedge_Q$.
  \begin{equation*}
    d(\xi \wedge_Q \eta) - (d\xi) \wedge_Q \eta - (-1)^{\abs{\xi}} \xi \wedge_Q d\eta = -\lambda H^{ji} \wedge (\partial_i \lrcorner \xi) \wedge \nabla_j \eta +\lambda (-1)^{\abs{\xi}} H^{ij} \wedge \nabla_i \xi \wedge (\partial_j \lrcorner \eta),
  \end{equation*}
  where
  \begin{equation}
    H^{ij} = \frac{1}{4} \omega^{is} (\nabla_s T^j_{nm} - 2R^j_{nms}) dx^m \wedge dx^n.   \label{eq:H^ij}
  \end{equation}
\end{proposition}
For an undeformed $d$ to satisfy the Leibniz rule for $\wedge_Q$, 
we need to take a flat torsion free connection.
If a background connection $\nabla$ is a Levi-Civita connection, 
torsion free is satisfied, but in general, neither curvature nor torsion is zero.
Hence, we modify $\wedge_Q$ and define the product $\wedge_1$ :
\begin{equation*}
  \xi \wedge_1 \eta = \xi \wedge_Q \eta + \lambda (-1)^{\abs{\xi} + 1} H^{ij} \wedge (\partial_i \lrcorner \xi) \wedge (\partial_j \lrcorner \eta).
\end{equation*}
If a background connection $\nabla$ is Poisson compatible, then $d$ satisfies the Leibniz rule for $\wedge_1$.
Therefore, we can define a differential calculus $\Omega^n (A_{\lambda})$ by $(\wedge_1, d)$.
In the following we call $\wedge_1$ the quantum wedge product.
\begin{remark}
  A quantum wedge product $\wedge_1$ is a multilinear map for $*$-products, i.e.
  \begin{equation*}
    (a * \xi) \wedge_1 (\eta * b) = a * (\xi \wedge_1 \eta) * b
  \end{equation*}
  for $a, b \in A_{\lambda}$ and $\xi, \eta \in \Omega^1(A_{\lambda})$.
  However, $\wedge_1$ is not a multilinear map for usual commutative products, i.e.
  \begin{equation*}
    (a \xi) \wedge_1 (\eta b) \neq a (\xi \wedge_1 \eta) b.
  \end{equation*}
  We will see this in the example of the upper half-plane below.
\end{remark}

Now that we have defined a quantum wedge product, 
we rewrite a torsion \eqref{eq:classicaltorsion} on a bimodule using quantization.
\begin{equation}
  T_{\nabla_Q} = d - \wedge_1 \nabla_Q,   \label{eq:quantumtorsion}
\end{equation}
where $\nabla_Q$ is the quantized connection $\nabla_{Q(E)}$ with $E = \Omega^1$, also used below.
Compared to the classical case \eqref{eq:classicaltorsion}, a wedge product and a connection are rewritten in quantized form.
Explicitly,
\begin{equation*}
  T_{\nabla_Q}(\xi) = T(\xi) - \frac{\lambda}{4} (\partial_j \lrcorner \nabla_i \xi) \omega^{is} (\nabla_s T^j_{nm}) dx^m \wedge dx^n
\end{equation*}
for $\xi \in \Omega^1(A_{\lambda})$.

\subsection{Semiquantization of Riemannian geometry}
we suppose that $(M, \omega, \nabla)$($\nabla$ is Poisson compatible) has additional structure $(g, \hat{\nabla})$ where $g$ is a Riemannian metric and $\hat{\nabla}$ is the Levi-Civita connection.
In this section we focus the case $E = \Omega^1(A_{\lambda})$.
We construct the quantum Riemannian metric and the quantum Levi-Civita connection defined in section \ref{subsubsec:riemmannianstronbimodule} by quantizing a Riemannian metric and a Levi-Civita connection.
Although \cite{name1} constructs quantum Levi-Civita connections for connections with torsion, 
in this paper we restrict $\hat{\nabla}$ to Levi-Civita connections (i.e. $\nabla = \hat{\nabla}$) and construct a quantum Levi-Civita connection.

\subsubsection{Semiquantization of Riemannian metric}
From definition \ref{def:quantumriemannianmetric}, a quantum Riemannian metric satisfied invertibility and symmetry.
To satisfy these two conditions, 
we impose the following condition on the background connection $\nabla$ :
\begin{equation}
  \nabla g = 0.   \label{eq:riemannianmetricquantizationcondition}
\end{equation}
We will see the necessity of this condition at the end of this section.

We quantize the Riemannian metric $g = g_{ij} dx^i \otimes_0 dx^j$ and check if it satisfies invertibility and symmetry.
By quantizing the tensor product $\otimes_0$ to $\otimes_1$, we have
\begin{equation}
  g_Q = q_{\Omega^1, \Omega^1}^{-1} (g) = g_{ij} dx^i \otimes_1 dx^j + \frac{\lambda}{2} \omega^{ij} g_{pm} \Gamma_{iq}^p \Gamma_{jn}^q dx^m \otimes_1 dx^n.    \label{eq:prequantumriemannmetric}
\end{equation}
Then, the following proposition holds.
\begin{proposition}   \label{prop:nablag=0}
  If we have $\nabla g = 0$, then we also have $\nabla_Q g_Q = 0$.
\end{proposition}
Here $\nabla_Q$ above is $\nabla_Q := Q(\nabla_{\Omega^1 \otimes \Omega^1}$).
Although $\nabla_Q$ is defined as $\nabla_{Q(\Omega^1)}$, 
we will use the same notation below for the quantized connection $\nabla_{\Omega^1 \otimes \Omega^1}$ on the tensor product bundle.

We now consider condition 2 of quantum Riemannian metric.
For the wedge product in the definition of quantum Riemannian metric, 
we use $\wedge_1$ obtained in section \ref{subsec:wedgeproduct}.
Computing the alternating part $\wedge_1(g_Q)$ of $g_Q$, 
we obtain
\begin{align*}
  \wedge_1 (g_1) & = g_{ij} dx^i \wedge_1 dx^j + \frac{\lambda}{2} \omega^{st} g_{li} \Gamma^l_{sj} \Gamma^j_{tk} dx^i \wedge_1 dx^k \\
                 & = g_{ij} dx^i \wedge dx^j + \frac{\lambda}{2} \omega^{st} \nabla_s (g_{ij} dx^i) \wedge \nabla_t dx^j + \lambda H^{ij} g_{ij} + \frac{\lambda}{2} \omega^{st} g_{li} \Gamma^l_{sj} \Gamma^j_{tk} dx^i \wedge_1 dx^k \\
                 & = \lambda \mathcal{R} - \frac{\lambda}{2} \omega^{st} (\partial_s g_{ij} - g_{lj} \Gamma^l_{si} - g_{li} \Gamma^l_{sj}) \Gamma^j_{tk} dx^i \wedge dx^k \\
                 & = \lambda \mathcal{R}
\end{align*}
from $\nabla g = 0$, and $g_Q$ is not symmetric in general.
$\mathcal{R}$ is called a generalised Ricci form and defined by
\begin{equation}
  \mathcal{R} = H^{ij} g_{ij}.   \label{eq:generalizedricci}
\end{equation}
Explicitly,
\begin{equation*}
  \mathcal{R} = \frac{1}{2} \mathcal{R}_{nm} dx^m \wedge dx^n, \quad \mathcal{R}_{nm} = \frac{1}{2} g_{ij} \omega^{is} (\nabla_s T_{nm}^j - 2R_{nms}^j).
\end{equation*}
\newpage
\noindent
Then we define $g_1$ with a $\lambda^1$ modification to $g_Q$ :
\begin{align}
  g_1 & = g_Q - \lambda q^{-1}_{\Omega^1, \Omega^1} \mathcal{R} \nonumber \\
      & = g_Q - \frac{\lambda}{4} g_{ij} \omega^{is} (\nabla_s T_{nm}^j - 2R_{nms}^j) dx^m \otimes_1 dx^n.    \label{eq:quantumriemannmetric}
\end{align}
In this case $\wedge_1 (g_1) = 0$.

Finally, we check if $g_1$ satisfies condition 1 of quantum Riemannian metric.
Here, the condition \eqref{eq:riemannianmetricquantizationcondition} is necessary.
From $\nabla g = 0$, the following lemma holds.
\begin{lemma}
  If $\nabla g = 0$, then $g_1$ is central in the quantised bimodule, i.e. $a * g_1 = g_1 * a$ for all $a \in A_{\lambda}$.
\end{lemma}
\begin{proof}
  In this proof we assume $g = g_{ij} dx^i \otimes_1 dx^j$.
  A $*$-product with the term $\lambda^1$ reduces to a commutative product $\cdot$, because we ignore more than $\lambda^2$ terms.
  Therefore, we have
  \begin{align*}
    a * (g_Q -\lambda q^{-1}_{\Omega^1, \Omega^1} \mathcal{R}) & = a * \qty(g + \frac{\lambda}{2} \omega^{ij} g_{pm} \Gamma^p_{iq} \Gamma^q_{jn} dx^m \otimes_1 dx^n - \lambda q^{-1}_{\Omega^1, \Omega^1} \mathcal{R}) \\
                                                               & = a * g + a \cdot \qty(\frac{\lambda}{2} \omega^{ij} g_{pm} \Gamma^p_{iq} \Gamma^q_{jn} dx^m \otimes_1 dx^n -\lambda q^{-1}_{\Omega^1, \Omega^1} \mathcal{R}) \\
                                                               & = a * g + \qty(\frac{\lambda}{2} \omega^{ij} g_{pm} \Gamma^p_{iq} \Gamma^q_{jn} dx^m \otimes_1 dx^n -\lambda q^{-1}_{\Omega^1, \Omega^1} \mathcal{R}) \cdot a \\
                                                               & = a * g + \qty(\frac{\lambda}{2} \omega^{ij} g_{pm} \Gamma^p_{iq} \Gamma^q_{jn} dx^m \otimes_1 dx^n -\lambda q^{-1}_{\Omega^1, \Omega^1} \mathcal{R}) * a.
  \end{align*}
  Thus, we need to check that $a * g = g * a$.
  From \eqref{eq:commutationrelation}, the commutator for a $*$-product of $a \in A_{\lambda}$ and $g \in \Omega^1(A_{\lambda}) \otimes_1 \Omega^1(A_{\lambda})$ is
  \begin{equation*}
    [a, g]_* = \lambda \omega^{ij} a_{,i} \nabla_j g.
  \end{equation*}
  If $\nabla g = 0$, then $[a, g]_* = 0$, hence $a * g = g * a$.
\end{proof}

Since $g_1$ is central, $g_1$ is invertible from the proposition \ref{prop:centralandinvertible}.

We construct a quantum Riemannian metric $g_1$ that satisfies the two conditions of quantum Riemannian metric.
From the above discussion, it can be seen that the condition $\nabla g = 0$ is relevant in both conditions 1 and 2 when constructing $g_1$.
Therefore, $\nabla g = 0$ is a necessary condition for constructing the quantum Riemannian metric $g_1$.

\subsubsection{Condition \texorpdfstring{$\nabla_Q$}{} is a quantum Levi-Civita connection}    \label{subsubsec:quantumlevicivitacondtion}
In section \ref{subsec:connectionquantization}, we quantized the connection $\nabla_E$ and constructed $\nabla_{Q(E)}$.
In this section, we investigate the conditions that the quantized connection $\nabla_Q$ is a quantum Levi-Civita connection.

If a background connection $\nabla$ is a Levi-Civita connection (i.e. $\nabla = \hat{\nabla}$), the following proposition holds.
\begin{proposition}   \label{prop:quantumlevicivita}
  If a background connection $\nabla$ is a Levi-Civita connection, then the following holds.
  \begin{enumerate}
    \item Poisson compatibility reduces to $\omega$ covariantly constant.
    \item $\nabla_Q$ is torsion free and $\mathcal{R} = -\frac{1}{2} g_{ij} \omega^{is} R^j_{nms} dx^m \wedge dx^n$ is closed.
    \item $\nabla_Q g_1 = 0$, i.e. $\nabla_Q$ is a quantum Levi-Civita connection for $g_1$, if and only if $\nabla \mathcal{R} = 0$.
  \end{enumerate}
\end{proposition}

\begin{remark}
  If the background connection is not a Levi-Civita connection (i.e. $\nabla \neq \hat{\nabla}$), 
  the Levi-Civita connection is induced by adding the contorsion tensor $S$ to $\nabla$, i.e. $\hat{\nabla} = \nabla + S$.
\end{remark}

From the above, we can obtain the condition that $\nabla_Q$ is a quantum Levi-Civita connection.

Finally, we define a weak quantum Levi-Civita connection using semiquantization.
If we take the background connection to be a Levi-Civita connection, 
the quantum torsion is zero from the proposition \ref{prop:quantumlevicivita}.
Therefore, \eqref{eq:cotorsion_2} is written as
\begin{equation}
  \mathrm{co}T_{\nabla_Q} = (\wedge \otimes_0 id) q^2 \nabla_Q g_1,   \label{eq:cotorsion_3}
\end{equation}
where $q^2 := q_{\Omega^1, \Omega^1 \otimes_0 \Omega^1}(id \otimes_1 q_{\Omega^1, \Omega^1})$.

\section{Example : noncommutative upper half-plane}    \label{sec:upperhalfplaneexample}
In this section we semiquantize the upper half-plane $M = \{ (x,y) \in \mathbb{R}^2 \;|\; y > 0 \}$.
As a result, something happens that is not seen in other examples.
In the latter of this section, 
we will examine the conditions that it occurs, 
and then we will semiquantize the upper half-plane by the generalized Poincar\'{e} metric.

\subsection{Classical data}    \label{subsec:classicalupperhalfplane}
We give the Riemannian metric in the upper half-plane by
\begin{equation}
  g = y^{-2} (dx \otimes dx + c^2 dy \otimes dy) \quad (c > 0).    \label{eq:riemannianmetricofupper}
\end{equation}
The case $c^2 = 1$ corresponds to the example in section 6.1 of \cite{name1}.
The case $c^2 = 2$ corresponds to the Fisher metric in the space of Gaussian distributions in information geometry.

The Christoffel symbols of the Levi-Civita connection for the Riemannian metric \eqref{eq:riemannianmetricofupper} are given by
\begin{align*}
  & \Gamma^1_{11} = \Gamma^1_{22} = 0, \quad \Gamma^1_{12} = \Gamma^1_{21} = -y^{-1} \\
  & \Gamma^2_{12} = \Gamma^2_{21} = 0, \quad \Gamma^2_{11} = c^{-2} y^{-1}, \quad \Gamma^2_{22} = -y^{-1}
\end{align*}
using \eqref{eq:christoffelsymbol}, where $x^1 = x$, $x^2 = y$.

Therefore, the covariant derivative of the 1-form $dx, dy$ is written as
\begin{align*}
  \nabla dx & = y^{-1} (dx \otimes dy + dy \otimes dx) \\
  \nabla dy & = -y^{-1} (c^{-2} dx \otimes dx - dy \otimes dy)
\end{align*}
from \eqref{eq:dx^icovariantlyderivative}.

Finding $\omega^{ij}$ that satisfies \eqref{eq:poissoncompatibility}
 for the Levi-Civita connection to be Poisson compatible, 
we have $\omega^{12} = y^2 (= -\omega^{21})$.
\newpage
\noindent
Also, from the Christoffel symbols obtained above, 
the nonzero components of the Riemann curvature tensor are
\begin{align*}
  R^1_{221} & = y^{-2}, \\
  R^1_{212} & = -R^1_{221}, \\
  R^2_{112} & = \frac{1}{c^2} y^{-2}, \\
  R^2_{121} & = -R^2_{112},
\end{align*}
from \eqref{eq:riemanncurvatureinlocal}.
Thus, we obtain a scalar curvature
\begin{equation*}
  R = g^{11} R^2_{121} + g^{22} R^1_{212} = -\frac{2}{c^2}.
\end{equation*}
Because of the negative constant scalar curvature, the upper half-plane is a negative constant curvature space.

Finally we compute the generalized Ricci form.
To compute the generalized Ricci form we need to compute the 2-forms $H^{ij}$.
In this paper we consider the case that the background connection is a Levi-Civita connection, 
thus, $H^{ij}$ is rewritten as
\begin{equation}
  H^{ij} = -\frac{1}{2} \omega^{is} R^j_{nms} dx^m \wedge dx^n    \label{eq:RiemannH^ij}
\end{equation}
from \eqref{eq:H^ij}.
Computing for $(i, j) = (1, 1)$, $(1, 2)$, $(2, 1)$ and $(2, 2)$, we obtain the following.
\begin{align*}
  H^{11} & = -\frac{1}{2} \omega^{12} R^1_{212} dx \wedge dy = -\frac{1}{2} y^2 (-y^{-2}) dx \wedge dy = \frac{1}{2} dx \wedge dy, \\
  H^{22} & = -\frac{1}{2} \omega^{21} R^2_{121} dy \wedge dx = -\frac{1}{2} y^2 \qty(-\frac{1}{c^2}y^{-2}) dx \wedge dy = \frac{1}{2c^2} dx \wedge dy, \\
  H^{12} & = H^{21} = 0.
\end{align*}
Therefore, the generalized Ricci form is
\begin{equation*}
  \mathcal{R} = H^{11}g_{11} + H^{22}g_{22} = y^{-2} dx \wedge dy
\end{equation*}
from \eqref{eq:generalizedricci}.

\subsection{Quantum data}
We calculate the quantum data using the classical data obtained in the previous section.
In the following, we restrict the vector bundle $E$ to the cotangent bundle $T^*M$ in section \ref{sec:semiclassical}.
We denote the connection $\nabla_E$ by $\nabla$ and $\nabla_{Ei}$ which is the covariant derivative along the vector field $\partial_i$ by $\nabla_i$ ($i = 1, 2$ since the upper half-plane is two-dimensional).
We also denote the quantized connection $\nabla_{Q(E)}$ by $\nabla_Q$ and the generalized braiding $\sigma_{Q(E)}$ by $\sigma_Q$.

\subsubsection{\texorpdfstring{$*$}{}-product}
First, the $*$-product for the functions $f, g$ is
\begin{equation*}
  f * g = fg + \frac{\lambda}{2} y^2 \qty(\pdv{f}{x} \pdv{g}{y} - \pdv{f}{y} \pdv{g}{x}).
\end{equation*}
For coordinate functions $x, y$,
\begin{alignat*}{2}
  & x * x = x^2,                        & & \quad y * y = y^2, \\
  & x * y = xy + \frac{\lambda}{2} y^2, & & \quad y * x = xy - \frac{\lambda}{2} y^2.
\end{alignat*}

Next, the $*$-product of the function $f$ and the differential form $dx, dy$ is
\begin{align*}
  f * dx & = fdx + \frac{\lambda}{2} y \qty(\pdv{f}{x} dx - \pdv{f}{y} dy), \\
  f * dy & = fdy + \frac{\lambda}{2} y \qty(\pdv{f}{x} dy + \frac{1}{c^2} \pdv{f}{y} dx).
\end{align*}
For coordinate functions $x, y$
\begin{align*}
  x * dx = xdx + \frac{\lambda}{2} ydx, & \quad y * dx = ydx - \frac{\lambda}{2} ydy, \\
  x * dy = xdy + \frac{\lambda}{2} ydy, & \quad y * dy = ydy + \frac{\lambda}{2c^2} ydx.
\end{align*}
Therefore, commutation relation of $x, y$ and $dx, dy$ is
\begin{align*}
  [x, dx]_* = \lambda y dx = c^2 [y, dy]_*, \\
  [x, dy]_* = \lambda y dy = -[y, dx]_*.
\end{align*}

\subsubsection{Quantum wedge product \texorpdfstring{$\wedge_1$}{}}
Using the 2-forms $H^{ij}$, the wedge product $dx \wedge_1 dx$ is
\begin{align}
  dx \wedge_1 dx & = dx \wedge_Q dx + \lambda H^{11} \nonumber \\
                 & = dx \wedge dx + \frac{\lambda}{2} \omega^{12} (\nabla_1 dx \wedge \nabla_2 dx - \nabla_2 dx \wedge \nabla_1 dx) + \frac{\lambda}{2} dx \wedge dy \nonumber \\
                 & = \frac{\lambda}{2} y^2 (y^{-1} dy \wedge y^{-1} dx - y^{-1} dx \wedge y^{-1} dy) + \frac{\lambda}{2} dx \wedge dy \nonumber \\
                 & = - \frac{\lambda}{2} dx \wedge dy.   \label{eq:dxwedge_1dx}
\end{align}
Similarly, $dx \wedge_1 dy$, $dy \wedge_1 dx$, $dy \wedge_1 dy$ are
\begin{align*}
  dx \wedge_1 dy & = dx \wedge dy, \\
  dy \wedge_1 dx & = dy \wedge dx, \\
  dy \wedge_1 dy & = - \frac{\lambda}{2c^2} dx \wedge dy.
\end{align*}
The results for $dx \wedge_1 dy$ and $dy \wedge_1 dx$ are equivalent to the classical case, respectively.
On the other hand, $dx \wedge_1 dx$ and $dy \wedge_1 dy$ are different from the classical case, respectively.
In classical case, the wedge product has graded commutativity ($\xi \wedge \eta = -\eta \wedge \xi$), so $dx \wedge dx$ and $dy \wedge dy$ are zero.
However, we do not assume graded commutativity for $\wedge_1$ (see Remark \ref{rem:ordercommutatibity}), so there is no problem that the result of $\wedge_1$ product of the same is not zero.

Also, for $\lambda \to 0$, $dx \wedge_1 dx$ and $dy \wedge_1 dy$ are zero, equal to the classical case.

\subsubsection{Quantum Riemannian metric \texorpdfstring{$g_1$}{}}
From \eqref{eq:quantumriemannmetric}, the quantum Riemannian metric can be rewritten as
\begin{equation*}
  g_1 = g_Q + \frac{\lambda}{2} g_{ij} \omega^{is} R_{nms}^j dx^m \otimes_1 dx^n
\end{equation*}
because the torsion tensor is zero since the background connection is a Levi-Civita connection.
Here $g_Q$ is written as
\begin{equation*}
  g_Q = g_{ij} dx^i \otimes_1 dx^j + \frac{\lambda}{2} \omega^{ij} g_{pm} \Gamma_{iq}^p \Gamma_{jn}^q dx^m \otimes_1 dx^n
\end{equation*}
from \eqref{eq:prequantumriemannmetric}.
First, we compute $g_Q$.
The second term $\omega^{ij} g_{ma} \Gamma^a_{ib} \Gamma^b_{jn}$ of $g_Q$ is
\begin{align*}
  \omega^{ij} g_{ma} \Gamma^a_{ib} \Gamma^b_{jn} & = y^2 g_{ma} \Gamma^a_{1b} \Gamma^b_{2n} - y^2 g_{ma} \Gamma^a_{2b} \Gamma^b_{1n} \\
                                                 & = y^2 g_{ma} \{ \Gamma^a_{1b} (-\delta^b_n y^{-1}) - (-\delta^a_b y^{-1}) \Gamma^b_{1n} \} \\
                                                 & = -y g_{ma} (\Gamma^a_{1b} \delta^b_n - \delta^a_b \Gamma^b_{1n}) \\
                                                 & = -y g_{ma} (\Gamma^m_{1b} \delta^b_n - \delta^m_b \Gamma^b_{1n}).
\end{align*}
When we compute in the cases $m = n$ and $m \neq n$, we have
\begin{align*}
  & m = n \ \text{case} \qquad \Gamma^m_{1m} - \Gamma^m_{1m} = 0, \\
  & m \neq n \ \text{case} \qquad \Gamma^m_{1n} - \Gamma^m_{1n} = 0.
\end{align*}
Thus, the $\lambda^1$ terms of $g_Q$ are zero.

We compute the alternating part of $g_Q$ and check that it is $\lambda \mathcal{R}$.
Acting $\wedge_1$ on $g_Q$ yields
\begin{equation*}
  \wedge_1 (g_Q) = \wedge_1 (y^{-2} (dx \otimes_1 dx + dy \otimes_1 dy)),
\end{equation*}
and computing $\wedge_1 (y^{-2} (dx \otimes_1 dx))$ yields
\begin{align}
  \wedge_1 (y^{-2} (dx \otimes_1 dx)) & = (y^{-2} dx) \wedge_1 dx \nonumber \\
                                      & = \frac{\lambda}{2} y^2 (\nabla_1 (y^{-2} dx) \wedge \nabla_2 dx - \nabla_2 (y^{-2} dx) \wedge \nabla_1 dx) + \frac{\lambda}{2} y^{-2} dx \wedge dy \nonumber \\
                                      & \nonumber \\
                                      & = \frac{\lambda}{2} y^2 \{ y^{-2} (y^{-1}dy \wedge y^{-1} dx) - (-2y^{-3}dx + y^{-2} y^{-1}dx) \wedge y^{-1}dy \} \nonumber \\
                                      & \quad + \frac{\lambda}{2} y^{-2} dx \wedge dy \nonumber \\
                                      & = \frac{\lambda}{2} y^{-2} dx \wedge dy \nonumber \\
                                      & = \frac{\lambda}{2} \mathcal{R}.   \label{eq:y-1dxwedge_1dx}
\end{align}
Similarly, $\wedge_1(y^{-2} (dy \otimes_1 dy)) = \frac{\lambda}{2} \mathcal{R}$.
Thus, we can confirm that the alternating part of $g_Q$ can be written in $\lambda \mathcal{R}$.

\begin{remark}
  As we mentioned in section \ref{subsec:wedgeproduct}, $\wedge_1$ is not a multilinear map for usual commutative products.
  In the upper half-plane example, they do not coincide, \eqref{eq:dxwedge_1dx} multiplied by $y^{-2}$ and \eqref{eq:y-1dxwedge_1dx}.
\end{remark}

From the above, the quantum Riemannian metric $g_1$ is
\begin{equation*}
  g_1 = y^{-2} (dx \otimes_1 dx + c^2 dy \otimes_1 dy) - \lambda y^{-2} (dx \otimes_1 dy - dy \otimes_1 dx)
\end{equation*}
from \eqref{eq:quantumriemannmetric}.

Unlike the Riemannian metric in the classical case, the terms $dx \otimes_1 dy$ and $dy \otimes_1 dx$ appear.
Also, in $\lambda \to 0$, $g_1$ is reduced to the classical form \eqref{eq:riemannianmetricofupper}.

\subsubsection{Quantised connection \texorpdfstring{$\nabla_Q$}{}}
In the general case, using \eqref{eq:nabla_QE} we compute $\nabla_Q dx^i$.
\begin{align}
  \nabla_Q dx^i & = \qty(q^{-1}_{\Omega^1, \Omega^1} \nabla_{\Omega^1} - \frac{\lambda}{2} \omega^{sj} dx^m \otimes_1 [\nabla_m, \nabla_j] \nabla_s) dx^i \nonumber \\
                & = q^{-1}_{\Omega^1, \Omega^1}(-\Gamma^i_{mn} dx^m \otimes_0 dx^n) + \frac{\lambda}{2} \omega^{sj} dx^m \otimes_1 [\nabla_m, \nabla_j] \Gamma^i_{sk} dx^k \nonumber \\
                & = -\Gamma^i_{mn} dx^m \otimes_1 dx^n + \frac{\lambda}{2} \omega^{sj} \nabla_s (\Gamma^i_{kt} dx^k) \otimes_1 \nabla_j dx^t \nonumber \\
                & \quad + \frac{\lambda}{2} \omega^{sj} dx^m \otimes_1 \Gamma^i_{sk} [\nabla_m, \nabla_j] dx^k \nonumber \\
                & = -\Gamma^i_{mn} dx^m \otimes_1 dx^n - \frac{\lambda}{2} \omega^{sj} (\partial_s \Gamma^i_{kt} dx^k - \Gamma^i_{kt} \Gamma^k_{sm} dx^m) \otimes_1 \Gamma^t_{jn} dx^n \nonumber \\
                & \quad + \frac{\lambda}{2} \omega^{sj} \Gamma^i_{sk} R^k_{nmj} dx^m \otimes_1 dx^n \nonumber \\
                & = -\Gamma^i_{mn} dx^m \otimes_1 dx^n - \frac{\lambda}{2} \omega^{sj} (\partial_s \Gamma^i_{mk} \Gamma^k_{jn} - \Gamma^i_{kt} \Gamma^k_{sm} \Gamma^t_{jn}) dx^m \otimes_1 dx^n \nonumber \\
                & \quad + \frac{\lambda}{2} \omega^{sj} \Gamma^i_{sk} R^k_{nmj} dx^m \otimes_1 dx^n \nonumber \\
                & = - \qty(\Gamma^i_{mn} + \frac{\lambda}{2} \omega^{sj} (\partial_s \Gamma^i_{mk} \Gamma^k_{jn} - \Gamma^i_{kt} \Gamma^k_{sm} \Gamma^t_{jn} - \Gamma^i_{sk} R^k_{nmj})) dx^m \otimes_1 dx^n.    \label{eq:nabla_Qdx}
\end{align}

\begin{remark}
  Compared to \eqref{eq:nabla_Qdx}, we think the last term $\Gamma^i_{jk} R^k_{nms}$ in (5,2) of \cite{name1} is a typo.
\end{remark}

Next, in the upper half-plane case we compute the terms $\lambda^1$ for $(m, n) = (1, 1)$, $(1, 2)$, $(2, 1)$, $(2, 2)$.
First, for $\nabla_Q dx$
\begin{subequations}
  \begin{align*}
    \intertext{$(m, n) = (1, 1)$ case}
    & \frac{\lambda}{2} \omega^{12} (0 - \Gamma^1_{21} \Gamma^2_{11} \Gamma^1_{21} - \Gamma^1_{12} R^2_{112}) + \frac{\lambda}{2} \omega^{21} (\Gamma^1_{12, 2} \Gamma^2_{11} - \Gamma^1_{12} \Gamma^1_{21} \Gamma^2_{11} - \underbrace{\Gamma^1_{2k} R^k_{111}}_{= 0}) \\
    & = \frac{\lambda}{2} y^2 \qty(-\Gamma^1_{21} \Gamma^2_{11} \Gamma^1_{21} + \Gamma^1_{12} \Gamma^1_{21} \Gamma^2_{11} + \frac{1}{c^2} y^{-1} y^{-2} - \frac{1}{c^2} y^{-2} y^{-1}) \\
    & = 0 \\
    \intertext{$(m, n) = (1, 2)$ case}
    & \frac{\lambda}{2} \omega^{12} (0 - \underbrace{\Gamma^1_{kt} \Gamma^k_{11} \Gamma^t_{22}}_{= 0} - \underbrace{\Gamma^1_{1k} R^k_{212}}_{= 0}) + \frac{\lambda}{2} \omega^{21} (\underbrace{\Gamma^1_{1k, 2} \Gamma^k_{12}}_{= 0} - \underbrace{\Gamma^1_{kt} \Gamma^k_{21} \Gamma^t_{12}}_{= 0} - \underbrace{\Gamma^1_{2k} R^k_{211}}_{= 0}) = 0 \\
    \intertext{$(m, n) = (2, 1)$ case}
    & \frac{\lambda}{2} \omega^{12} (0 - \underbrace{\Gamma^1_{kt} \Gamma^k_{12} \Gamma^t_{21}}_{= 0} - \underbrace{\Gamma^1_{1k} R^k_{122}}_{= 0}) + \frac{\lambda}{2} \omega^{21} (\underbrace{\Gamma^1_{2k, 2} \Gamma^k_{11}}_{= 0} - \underbrace{\Gamma^1_{kt} \Gamma^k_{22} \Gamma^t_{11}}_{= 0} - \underbrace{\Gamma^1_{2k} R^k_{121}}_{= 0}) = 0 \\
    \intertext{$(m, n) = (2, 2)$ case}
    & \frac{\lambda}{2} \omega^{12} (0 - \Gamma^1_{12} \Gamma^1_{12} \Gamma^2_{22} - \underbrace{\Gamma^1_{1k} R^k_{222}}_{= 0}) + \frac{\lambda}{2} \omega^{21} (\Gamma^1_{21, 2} \Gamma^1_{12} - \Gamma^1_{21} \Gamma^2_{22} \Gamma^1_{12} - \Gamma^1_{21} R^1_{221}) \\
    & = \frac{\lambda}{2} y^2 (-\Gamma^1_{12} \Gamma^1_{12} \Gamma^2_{22} + \Gamma^1_{21} \Gamma^2_{22} \Gamma^1_{12} + y^{-2} y^{-1} - y^{-1} y^{-2}) \\
    & = 0
  \end{align*}
\end{subequations}
Next, for $\nabla_Q dy$
\begin{subequations}
  \begin{align*}
    \intertext{$(m, n) = (1, 1)$ case}
    & \frac{\lambda}{2} \omega^{12} (0 - \underbrace{\Gamma^2_{21} \Gamma^2_{11} \Gamma^1_{21}}_{= 0} - \underbrace{\Gamma^1_{2k} R^k_{112}}_{= 0}) + \frac{\lambda}{2} \omega^{21} (\underbrace{\Gamma^2_{11, 2} \Gamma^1_{11}}_{= 0} - \underbrace{\Gamma^2_{12} \Gamma^1_{21} \Gamma^2_{11}}_{= 0} - \underbrace{\Gamma^2_{2k} R^k_{111}}_{= 0}) = 0 \\
    \intertext{$(m, n) = (1, 2)$ case}
    & \frac{\lambda}{2} \omega^{12} (0 - \Gamma^2_{22} \Gamma^2_{11} \Gamma^2_{22} - \Gamma^2_{11} R^1_{212}) + \frac{\lambda}{2} \omega^{21} (\Gamma^2_{11, 2} \Gamma^1_{12} - \Gamma^2_{11} \Gamma^1_{21} \Gamma^1_{12} - \underbrace{\Gamma^2_{2k} R^k_{211}}_{= 0}) \\
    & = \frac{\lambda}{2} y^2 \qty(-\Gamma^2_{22} \Gamma^2_{11} \Gamma^2_{22} + \Gamma^2_{11} \Gamma^1_{21} \Gamma^1_{12} + \frac{1}{c^2} y^{-1} y^{-2} - \frac{1}{c^2} y^{-2} y^{-1}) \\
    & = 0 \\
    \intertext{$(m, n) = (2, 1)$ case}
    & \frac{\lambda}{2} \omega^{12} (0 - \Gamma^2_{11} \Gamma^1_{12} \Gamma^1_{21} - \underbrace{\Gamma^2_{1k} R^k_{122}}_{= 0}) + \frac{\lambda}{2} \omega^{21} (\Gamma^2_{22, 2} \Gamma^2_{11} - \Gamma^2_{22} \Gamma^2_{22} \Gamma^2_{11} - \Gamma^2_{22} R^2_{121}) \\
    & = \frac{\lambda}{2} y^2 \qty(-\Gamma^2_{11} \Gamma^1_{12} \Gamma^1_{21} + \Gamma^2_{22} \Gamma^2_{22} \Gamma^2_{11} - \frac{1}{c^2} y^{-2} y^{-1} + \frac{1}{c^2} y^{-1} y^{-2}) \\
    & = 0 \\
    \intertext{$(m, n) = (2, 2)$ case}
    & \frac{\lambda}{2} \omega^{12} (0 - \underbrace{\Gamma^2_{12} \Gamma^1_{12} \Gamma^2_{22}}_{= 0} - \underbrace{\Gamma^2_{1k} R^k_{222}}_{= 0}) + \frac{\lambda}{2} \omega^{21} (\underbrace{\Gamma^2_{22, 2} \Gamma^2_{12}}_{= 0} - \underbrace{\Gamma^2_{21} \Gamma^2_{22} \Gamma^1_{12}}_{= 0} - \underbrace{\Gamma^2_{2k} R^k_{221}}_{= 0}) = 0
  \end{align*}
\end{subequations}
\newpage
From the above, $\nabla_Q dx^i$ in the upper half-plane is
\begin{align*}
  \nabla_Q dx & = y^{-1} (dx \otimes_1 dy + dy \otimes_1 dx), \\
  \nabla_Q dy & = y^{-1} (-c^{-2} dx \otimes_1 dx + dy \otimes_1 dy).
\end{align*}

The fact that the $\lambda^1$ terms are each zero indicates that 
in the upper half-plane example $\nabla_Q dx^i$ has the same form 
as classical one (but the tensor product is quantized).

From \eqref{eq:braiding}, we compute $\sigma_Q$ for $\eta, \xi \in \Omega^1(A_{\lambda})$ :
\begin{equation*}
  \sigma_{Q}(\eta \otimes_1 \xi) = \xi \otimes_1 \eta + \lambda \omega^{ij} \nabla_j \xi \otimes_1 \nabla_i \eta + \lambda \omega^{ij} \xi_j \eta_s R^s_{nki} dx^k \otimes_1 dx^n.
\end{equation*}
\begin{remark}
  In the calculation of $\sigma_Q (\eta \otimes_1 \xi)$ shown on page 464 in \cite{name1}, 
  we think the sign of the third term is a typo.
\end{remark}
Here, we compute the terms needed to compute $\nabla_Q g_Q$ in the next section.
\begin{align}
  \sigma_Q (y^{-2} dx \otimes_1 y^{-1} dx) & = y^{-3} dx \otimes_1 dx + \lambda y^{-3} (dx \otimes_1 dy + dy \otimes_1 dx),    \label{eq:sigma_Qdxotimesdx} \\
  \sigma_Q (y^{-2} dx \otimes_1 y^{-1} dy) & = y^{-3} dy \otimes_1 dx + \lambda y^{-3} (-c^{-2} dx \otimes_1 dx + dy \otimes_1 dy), \nonumber \\
  \sigma_Q (y^{-2} dy \otimes_1 y^{-1} dx) & = y^{-3} dx \otimes_1 dy + \lambda y^{-3} (-c^{-2} dx \otimes_1 dx + dy \otimes_1 dy), \nonumber \\
  \sigma_Q (y^{-2} dy \otimes_1 y^{-1} dy) & = y^{-3} dy \otimes_1 dy - \lambda c^{-2} y^{-3} (dx \otimes_1 dy + dy \otimes_1 dx). \nonumber
\end{align}
On the other hand we compute $y^{-3} \sigma_Q (dx \otimes_1 dx)$ :
\begin{equation*}
  y^{-3} \sigma_Q (dx \otimes_1 dx) = y^{-3} dx \otimes_1 dx + \lambda y^{-3} (2dx \otimes_1 dy - dy \otimes_1 dx).
\end{equation*}
Comparing with \eqref{eq:sigma_Qdxotimesdx}, we see that $\sigma_Q$ is not a bimodule map with respect to a commutative product(as mentioned Remark \ref{rem:braiding}).

Finally we check if $\nabla_Q$ above is a quantum Levi-Civita connection for $g_1$.
Before computing $\nabla_Q g_1$ we compute $\nabla_Q g_Q$.
From \eqref{eq:tensorbimoduleconnection}, $\nabla_Q g_Q$ can be written as
\begin{align*}
  \nabla_Q g_Q & = \nabla_Q \{ y^{-2}(dx \otimes_1 dx + c^2 dy \otimes_1 dy) \} \\
               & = \nabla_Q (y^{-2} dx) \otimes_1 dx + (\sigma_Q \otimes_1 id) (y^{-2} dx \otimes_1 \nabla_Q dx) \\
               & \quad + c^2 \nabla_Q (y^{-2} dy) \otimes_1 dy + c^2 (\sigma_Q \otimes_1 id) (y^{-2} dy \otimes_1 \nabla_Q dy).
\end{align*}
From here we calculate each term respectively.
\begin{subequations}
  \begin{align*}
    \intertext{The first term}
    & \nabla_Q (y^{-2} dx) \otimes_1 dx \\
    & = (dy^{-2} \otimes_1 dx + y^{-2} \nabla_Q dx) \otimes_1 dx \\
    & = -2y^{-3} dy \otimes_1 dx \otimes_1 dx + y^{-2} \{ y^{-1} (dx \otimes_1 dy + dy \otimes_1 dx) \} \otimes_1 dx \\
    & = -y^{-3} dy \otimes_1 dx \otimes_1 dx + y^{-3} dx \otimes_1 dy \otimes_1 dx \\
    \intertext{The second term}
    & (\sigma_Q \otimes_1 id) (y^{-2} dx \otimes_1 \nabla_Q dx) \\
    & = (\sigma_Q \otimes_1 id) ( y^{-2} dx \otimes_1 \{ y^{-1} (dx \otimes_1 dy + dy \otimes_1 dx) \}) \\
    & = \sigma_Q(y^{-2} dx \otimes_1 dx) \otimes_1 dy + \sigma_Q(y^{-2} dx \otimes_1 y^{-1} dy) \otimes_1 dx \\
    & = \{ y^{-3} dx \otimes_1 dx + \lambda y^{-3} (dx \otimes_1 dy + dy \otimes_1 dx) \} \otimes_1 dy \\
    & \quad + \{ y^{-3} dy \otimes_1 dx + \lambda y^{-3} (-c^{-2} dx \otimes_1 dx + dy \otimes_1 dy) \} \otimes_1 dx \\
    & = y^{-3} dx \otimes_1 dx \otimes_1 dy + y^{-3} dy \otimes_1 dx \otimes_1 dx \\
    & \quad + \lambda y^{-3} (dx \otimes_1 dy \otimes dy + dy \otimes_1 dx \otimes_1 dy) \\
    & \quad + \lambda y^{-3} (- c^{-2} dx \otimes_1 dx \otimes_1 dx + dy \otimes_1 dy \otimes_1 dx) \\
    \intertext{The third term}
    & c^2 \nabla_Q (y^{-2} dy) \otimes_1 dy \\
    & = c^2 (dy^{-2} \otimes_1 dy + y^{-2} \nabla_Q dy) \otimes_1 dy \\
    & = -2 c^2 y^{-3} dy \otimes_1 dy \otimes_1 dy + c^2 y^{-2} \{ y^{-1} (-c^{-2} dx \otimes_1 dx + dy \otimes_1 dy) \} \otimes_1 dy \\
    & = -c^2 y^{-3} dy \otimes_1 dy \otimes_1 dy - y^{-3} dx \otimes_1 dx \otimes_1 dy \\
    \intertext{The fourth term}
    & c^2 (\sigma_Q \otimes_1 id) (y^{-2} dy \otimes_1 \nabla_Q dy) \\
    & = c^2 (\sigma_Q \otimes_1 id) (y^{-2} dy \otimes_1 \{ y^{-1} (-c^{-2} dx \otimes_1 dx + dy \otimes_1 dy) \} ) \\
    & = -\sigma_Q(y^{-2} dy \otimes_1 y^{-1} dx) \otimes_1 dx + c^2 \sigma_Q(y^{-2} dy \otimes_1 y^{-1} dy) \otimes_1 dy \\
    & = - \{ y^{-3} dx \otimes_1 dy + \lambda y^{-3} (-c^{-2} dx \otimes_1 dx + dy \otimes_1 dy) \} \otimes_1 dx \\
    & \quad + c^2 \{ y^{-3} dy \otimes_1 dy - \lambda c^{-2} y^{-3} (dx \otimes_1 dy + dy \otimes_1 dx) \} \otimes_1 dy \\
    & = -y^{-3} dx \otimes_1 dy \otimes_1 dx + c^2 y^{-3} dy \otimes_1 dy \otimes_1 dy \\
    & \quad + \lambda y^{-3} (c^{-2} dx \otimes_1 dx \otimes_1 dx - dy \otimes_1 dy \otimes_1 dx) \\
    & \quad + \lambda y^{-3} (- dx \otimes_1 dy \otimes_1 dy - dy \otimes_1 dx \otimes_1 dy)
  \end{align*}
\end{subequations}
From the above, $\nabla_Q g_Q = 0$.

We calculate $\nabla_Q g_1$.
Since $\nabla_Q g_Q = 0$, we have
\begin{equation*}
  q^2(\nabla_Q g_1) = q^2(\nabla_Q g_Q - \lambda \nabla_Q (q^{-1}_{\Omega^1, \Omega^1} \mathcal{R})) = - \lambda \nabla \mathcal{R}.
\end{equation*}
Thus, we need to compute $\nabla \mathcal{R}$.
\begin{align*}
  \nabla \mathcal{R} & = \nabla \{ y^{-2} (dx \otimes dy - dy \otimes dx) \} \\
                     & = \nabla (y^{-2} dx) \otimes dy + (\sigma \otimes id) (y^{-2} dx \otimes \nabla dy) - \nabla (y^{-2} dy) \otimes dx - (\sigma \otimes id) (y^{-2} dy \otimes dx) \\
                     & = \{ -2y^{-3} dy \otimes dx + y^{-2} (y^{-1} dx \otimes dy + y^{-1} dy \otimes dx) \} \otimes dy \\
                     & \quad + (\sigma \otimes id) \{ y^{-2} dx \otimes (-c^{-2} y^{-1} dx \otimes dx + y^{-1} dy \otimes dy) \} \\
                     & \quad - \{ -2y^{-3} dy \otimes dy + y^{-2} (-c^{-2} y^{-1} dx \otimes dx + y^{-1} dy \otimes dy) \} \otimes dx \\
                     & \quad - (\sigma \otimes id) \{ y^{-2} dy \otimes (y^{-1} dx \otimes dy + y^{-1} dy \otimes dx) \} \\
                     & = -2y^{-3} dy \otimes dx \otimes dy + y^{-3} (dx \otimes dy + dy \otimes dx) \otimes dy \\
                     & \quad - c^{-2} y^{-3} dx \otimes dx \otimes dx + y^{-3} dy \otimes dx \otimes dy + 2y^{-3} dy \otimes dy \otimes dx \\
                     & \quad - y^{-3} (- c^{-2} dx \otimes dx + dy \otimes dy) \otimes dx - y^{-3} dx \otimes dy \otimes dy - y^{-3} dy \otimes dy \otimes dx \\
                     & = 0.
\end{align*}
As a result we get $\nabla_Q g_1 = 0$.
Thus, $\nabla_Q$ is a quantum Levi-Civita connection for the quantum Riemannian metric $g_1$.

\subsection{Main result 1 : Conditions that have the same form as classical one}    \label{subsec:classicalreduction}
When we calculated the quantum Riemannian metric $g_Q$ and the quantum Levi-Civita connection $\nabla_Q$ in the previous section, 
the $\lambda^1$ terms were zero and they had the same form as classical one.
In this section, we investigate the Riemannian metric in the upper half-plane by generalizing it as follows :
\begin{equation}
  g = E(y) dx \otimes dx + G(y) dy \otimes dy,    \label{eq:riemannianmetricEG}
\end{equation}
where $E, G$ are arbitrary functions and $E(y), G(y) > 0$ from the positive definiteness of $g$.

\begin{itembox}[l]{\textbf{Main result 1}}
  Suppose the Riemannian metric $g$ is of the form \eqref{eq:riemannianmetricEG}.
  The condition that the quantized metric $g_Q$ and the quantized connection $\nabla_Q$ 
  have the same form as in the classical case is 
  that $E(y), G(y) > 0$ satisfies either of the following.
  \begin{enumerate}
    \item $E(y)$ is constant function.
    \item $G(y) = AE(y)$\;($A$ is a positive constant).
  \end{enumerate}
\end{itembox}
From here in this section, we derive the main result 1.

First, from \eqref{eq:christoffelsymbol}, the Christoffel symbols of the Levi-Civita connection for the Riemannian metric $g$ given by \eqref{eq:riemannianmetricEG} are as follows.
\begin{alignat*}{3}
  \Gamma^1_{11} & = 0,                    \quad & \Gamma^1_{12} & = \Gamma^1_{21} = \frac{E'(y)}{2E(y)}, \quad & \Gamma^1_{22} & = 0, \\
  \Gamma^2_{11} & = -\frac{E'(y)}{2G(y)}, \quad & \Gamma^2_{12} & = \Gamma^2_{21} = 0,                   \quad & \Gamma^2_{22} & = \frac{G'(y)}{2G(y)},
\end{alignat*}
where $E'(y)$ means the derivative with $y$.
Using the Christoffel symbols obtained above, the nonzero components of the Riemann curvature tensor are
\begin{align*}
  R^1_{221} & = (\Gamma^1_{12})' + \Gamma^1_{12} \Gamma^1_{21} - \Gamma^2_{22} \Gamma^1_{12}, \\
  R^1_{212} & = -R^1_{221}, \\
  R^2_{112} & = -(\Gamma^2_{11})' + \Gamma^1_{12} \Gamma^2_{11} - \Gamma^2_{11} \Gamma^2_{22}, \\
  R^2_{121} & = -R^2_{112}.
\end{align*}

\subsubsection{Conditions for \texorpdfstring{$g_Q$}{}}    \label{subsubsec:g_Q}
From \eqref{eq:prequantumriemannmetric} $g_Q$ is written as
\begin{equation*}
  g_Q = g_{ij} dx^i \otimes_1 dx^j + \frac{\lambda}{2} \omega^{ij} g_{pm} \Gamma_{iq}^p \Gamma_{jn}^q dx^m \otimes_1 dx^n.
\end{equation*}
$g_Q$ has the same form as classical one means that the $\lambda^1$ terms are zero, i.e. for each $m, n$
\begin{equation*}
  \epsilon_{ij} \delta_{pm} \Gamma_{iq}^p \Gamma_{jn}^q = 0,
\end{equation*}
where $\epsilon_{ij}$ is completely anti­symmetric tensor which is $\epsilon_{12} = -\epsilon_{21} = 1$,$\epsilon_{11} = \epsilon_{22} = 0$.
And $\delta_{pm}$ is the Kronecker delta.
We compute the $\lambda^1$ terms for the cases $(m, n) = (1, 1)$, $(1, 2)$, $(2, 1)$ and $(2, 2)$ respectively.
\begin{subequations}
  \begin{align*}
    \intertext{$(m, n) = (1, 1)$ case}
    & \underbrace{\Gamma^1_{1q} \Gamma^q_{21}}_{= 0} - \underbrace{\Gamma^1_{2q} \Gamma^q_{11}}_{= 0} = 0 \\
    \intertext{$(m, n) = (1, 2)$ case}
    & \Gamma^1_{12} \Gamma^2_{22} - \Gamma^1_{21} \Gamma^1_{12} = \frac{E'(y)(-G(y) E'(y) + E(y) G'(y))}{4E(y)^2 G(y)} \\
    \intertext{$(m, n) = (2, 1)$ case}
    & \Gamma^2_{11} \Gamma^1_{21} - \Gamma^2_{22} \Gamma^2_{11} = \frac{E'(y)(-G(y) E'(y) + E(y) G'(y))}{4E(y) G(y)^2} \\
    \intertext{$(m, n) = (2, 2)$ case}
    & \underbrace{\Gamma^2_{1q} \Gamma^q_{22}}_{= 0} - \underbrace{\Gamma^2_{2q} \Gamma^q_{12}}_{= 0} = 0
  \end{align*}
\end{subequations}
In $(m, n) = (1, 2)$, $(2, 1)$ cases, $\lambda^1$ terms of $g_Q$ are not zero.
Both of these cases are of the form of functional multiples of the derivative of the function $\frac{G(y)}{E(y)}$.
In order for the terms $\lambda^1$ to be zero, it is necessary to be
\begin{equation*}
  E'(y) = 0 \quad \text{or} \quad \qty(\frac{G(y)}{E(y)})' = 0.
\end{equation*}
From the above, it can be seen that 
the condition for having the same form as the classical one 
is satisfied by either of the following.
\begin{enumerate}
  \item $E(y)$ is a constant function.
  \item $G(y) = A E(y)$\;($A$ is a positive constant).
\end{enumerate}

\subsubsection{Conditions for \texorpdfstring{$\nabla_Q dx^i$}{}}   \label{subsubsec:nabla_Q}
$\nabla_Q$ is written as
\begin{equation*}
  \nabla_Q dx^i = - \qty(\Gamma^i_{mn} + \frac{\lambda}{2} \omega^{sj} (\Gamma^i_{mk,s} \Gamma^k_{jn} - \Gamma^i_{kt} \Gamma^k_{sm} \Gamma^t_{jn} - \Gamma^i_{sk} R^k_{nmj})) dx^m \otimes_1 dx^n
\end{equation*}
from \eqref{eq:nabla_Qdx}.
$\nabla_Q$ has the same form as classical one means that the $\lambda^1$ terms in $g_Q$ are zero, i.e. for each $m, n$
\begin{equation*}
  \epsilon_{sj} (\Gamma^i_{mk,s} \Gamma^k_{jn} - \Gamma^i_{kt} \Gamma^k_{sm} \Gamma^t_{jn} - \Gamma^i_{sk} R^k_{nmj}) = 0.
\end{equation*}

We compute the $\lambda^1$ terms for the cases $(m, n) = (1, 1)$, $(1, 2)$, $(2, 1)$ and $(2, 2)$ respectively.
First, for $\nabla_Q dx$
\begin{subequations}
  \begin{align*}
    \intertext{$(m, n) = (1, 1)$ case}
    & (0 - \Gamma^1_{21} \Gamma^2_{11} \Gamma^1_{21} - \Gamma^1_{12} R^2_{112}) - (\Gamma^1_{12, 2} \Gamma^2_{11} - \Gamma^1_{12} \Gamma^1_{21} \Gamma^2_{11} - \underbrace{\Gamma^1_{2k} R^k_{111}}_{= 0}) \\
    & = \frac{E'(y)^2 (-G(y) E'(y) + E(y) G'(y))}{8E(y)^2 G(y)^2} \\
    \intertext{$(m, n) = (1, 2)$ case}
    & (0 - \underbrace{\Gamma^1_{kt} \Gamma^k_{11} \Gamma^t_{22}}_{= 0} - \underbrace{\Gamma^1_{1k} R^k_{212}}_{= 0}) - (\underbrace{\Gamma^1_{1k, 2} \Gamma^k_{12}}_{= 0} - \underbrace{\Gamma^1_{kt} \Gamma^k_{21} \Gamma^t_{12}}_{= 0} - \underbrace{\Gamma^1_{2k} R^k_{211}}_{= 0}) = 0 \\
    \intertext{$(m, n) = (2, 1)$ case}
    & (0 - \underbrace{\Gamma^1_{kt} \Gamma^k_{12} \Gamma^t_{21}}_{= 0} - \underbrace{\Gamma^1_{1k} R^k_{122}}_{= 0}) - (\underbrace{\Gamma^1_{2k, 2} \Gamma^k_{11}}_{= 0} - \underbrace{\Gamma^1_{kt} \Gamma^k_{22} \Gamma^t_{11}}_{= 0} - \underbrace{\Gamma^1_{2k} R^k_{121}}_{= 0}) = 0 \\
    \intertext{$(m, n) = (2, 2)$ case}
    & (0 - \Gamma^1_{12} \Gamma^1_{12} \Gamma^2_{22} - \underbrace{\Gamma^1_{1k} R^k_{222}}_{= 0}) - (\Gamma^1_{21, 2} \Gamma^1_{12} - \Gamma^1_{21} \Gamma^2_{22} \Gamma^1_{12} - \Gamma^1_{21} R^1_{221}) \\
    & = \frac{E'(y)^2 (-G(y) E'(y) + E(y) G'(y))}{8E(y)^3 G(y)}
  \end{align*}
\end{subequations}
Next, for $\nabla_Q dy$
\begin{subequations}
  \begin{align*}
    \intertext{$(m, n) = (1, 1)$ case}
    & (0 - \underbrace{\Gamma^2_{kt} \Gamma^k_{11} \Gamma^t_{21}}_{= 0} - \underbrace{\Gamma^1_{2k} R^k_{112}}_{= 0}) - (\underbrace{\Gamma^2_{1k, 2} \Gamma^k_{11}}_{= 0} - \underbrace{\Gamma^2_{kt} \Gamma^k_{21} \Gamma^t_{11}}_{= 0} - \underbrace{\Gamma^2_{2k} R^k_{111}}_{= 0}) = 0 \\
    \intertext{$(m, n) = (1, 2)$ case}
    & (0 - \Gamma^2_{22} \Gamma^2_{11} \Gamma^2_{22} - \Gamma^2_{11} R^1_{212}) - (\Gamma^2_{11, 2} \Gamma^1_{12} - \Gamma^2_{11} \Gamma^1_{21} \Gamma^1_{12} - \underbrace{\Gamma^2_{2k} R^k_{211}}_{= 0}) \\
    & = \frac{E'(y) G'(y) (-G(y) E'(y) + E(y) G'(y))}{8E(y) G(y)^3} \\
    \intertext{$(m, n) = (2, 1)$ case}
    & (0 - \Gamma^2_{11} \Gamma^1_{12} \Gamma^1_{21} - \underbrace{\Gamma^2_{1k} R^k_{122}}_{= 0}) - (\Gamma^2_{22, 2} \Gamma^2_{11} - \Gamma^2_{22} \Gamma^2_{22} \Gamma^2_{11} - \Gamma^2_{22} R^2_{121}) \\
    & = \frac{1}{8E(y)^2 G(y)^3} \left\{ G(y)^2 E'(y)^3 - 2E(y)^2 E'(y) G'(y)^2 \right. \\
    & \quad \left. + E(y) G(y) \left(E'(y)^2 G'(y) - 2E(y) G'(y) E''(y) + 2E(y) E'(y) G''(y) \right) \right\} \\
    \intertext{$(m, n) = (2, 2)$ case}
    & (0 - \underbrace{\Gamma^2_{kt} \Gamma^k_{12} \Gamma^t_{22}}_{= 0} - \underbrace{\Gamma^2_{1k} R^k_{222}}_{= 0}) - (\underbrace{\Gamma^2_{2k, 2} \Gamma^k_{12}}_{= 0} - \underbrace{\Gamma^2_{kt} \Gamma^k_{22} \Gamma^t_{12}}_{= 0} - \underbrace{\Gamma^2_{2k} R^k_{221}}_{= 0}) = 0
  \end{align*}
\end{subequations}

In summary, the condition for the terms $\lambda^1$ to be zero for each $m, n$ is that the following is satisfied :
\begin{align}
  & \frac{E'(y)^2 (-G(y) E'(y) + E(y) G'(y))}{8E(y)^2 G(y)^2} = 0, \label{eq:dx11} \\
  & \frac{E'(y)^2 (-G(y) E'(y) + E(y) G'(y))}{8E(y)^3 G(y)} = 0, \label{eq:dx22} \\
  & \frac{E'(y) G'(y) (-G(y) E'(y) + E(y) G'(y))}{8E(y) G(y)^3} = 0, \label{eq:dy12} \\
  & \frac{1}{8E(y)^2 G(y)^3} \left\{ G(y)^2 E'(y)^3 - 2E(y)^2 E'(y) G'(y)^2 \right. \nonumber \\
  & \quad \left. + E(y) G(y) \left(E'(y)^2 G'(y) - 2E(y) G'(y) E''(y) + 2E(y) E'(y) G''(y) \right) \right\} = 0. \label{eq:dy21}
\end{align}
\eqref{eq:dx11}, \eqref{eq:dx22}, and \eqref{eq:dy12} have a form similar to that of $g_Q$, 
and the conditions that they are zero can be rewritten as follows :
\begin{equation}
  E'(y) = 0 \quad \text{or} \quad \qty(\frac{G(y)}{E(y)})' = 0.   \label{eq:dx11dx22dy12condition}
\end{equation}
This condition is also satisfy \eqref{eq:dy21}.
Thus, $E(y)$ and $G(y)$ satisfying this condition are the same as for $g_Q$, 
where $E(y)$ is a constant function or $G(y) = AE(y)$\;($A$ is a positive constant).

Here we will investigate the \eqref{eq:dy21} in more detail.
In fact, there exists a solution to \eqref{eq:dy21}, which is a more generalization of \eqref{eq:dx11dx22dy12condition}.
\begin{lemma}
  Let $G(y) = A(y) E(y)$, where $A(y)$ is an arbitrary positive function.
  In this case, it is equivalent that \eqref{eq:dy21} and $G(y) = c_2 e^{2c_1 \sqrt{E(y)}} E(y)$, 
  where $c_1$ is an arbitrary constant, 
  and $c_2$ is an arbitrary positive constant since $G(y)$ is positive.
\end{lemma}
If $c_1 = 0$, then $G(y) = c_2 E(y)$, so the conditions \eqref{eq:dx11}, \eqref{eq:dx22} and \eqref{eq:dy12} are satisfied.
However, if $c_1 \neq 0$, then $G(y) = A(y) E(y)$, so the conditions \eqref{eq:dx11}, \eqref{eq:dx22}, and \eqref{eq:dy12} are not satisfied.
This means that the $\lambda^1$ terms for $(m, n) = (2, 1)$ in $\nabla_Q dy$ are zero, 
while the $\lambda^1$ terms in $\nabla_Q dx$ with $(m, n) = (1, 1)$, $(2, 2)$ and $\nabla_Q dy$ with $(m, n) = (1, 2)$ are not zero.

In summary, we have
\begin{equation*}
  \text{$E(y)$ is a constant function} \quad \text{or} \quad G(y) = A E(y)\;\text{($A$ is a positive constant)}
\end{equation*}
as the condition that $\nabla_Q dx^i$ has the same form as classical one.
From the above we obtain main result 1.

From the main result 1, the Riemannian metric
\begin{equation*}
  g = y^{-2} (dx \otimes dx + c^2 dy \otimes dy)
\end{equation*}
in the upper half-plane is of the form $G(y) = A E(y)$\;($A$ is a positive constant), 
which means that $g_Q$ and $\nabla_Q dx^i$ have the same form as classical one.
We also see that there is a (positive) function degree of freedom of $E(y)$ in having the same form as classical one.

In this section, the following two expressions appeared as the relation between $G(y)$ and $E(y)$ :
\begin{enumerate}
  \item $G(y) = A E(y)$,
  \item $G(y) = c_2 e^{2 c_1 \sqrt{E(y)}} E(y)$.
\end{enumerate}
In section \ref{subsec:G(y)=AE(y)}, 
we will study the conditions that the quantized connection $\nabla_Q$ is 
a quantum Levi-Civita connection when the first condition holds.
In section \ref{subsec:G(y)=A(y)E(y)}, 
we compute the quantum Riemannian metric $g_1$ and 
the quantized connection $\nabla_Q$ for the second condition and 
see whether $\nabla_Q$ is a quantum Levi-Civita connection.

\subsection{Main result 2 : \texorpdfstring{$G(y) = A E(y)$}{} case}    \label{subsec:G(y)=AE(y)}
If the Riemannian metric in the upper half-plane is written
\begin{equation}
  g = E(y) dx \otimes dx + A E(y) dy \otimes dy,   \label{eq:RiemannianmetricG(y)=AE(y)}
\end{equation}
then the quantized metric $g_Q$ and the quantized connection $\nabla_Q$ have the same form as classical one from section \ref{subsec:classicalreduction}.
In this section, we further investigate the conditions that 
the quantized connection $\nabla_Q$ is a quantum Levi-Civita connection for the quantum Riemannian metric $g_1$.
To do so, we first need to compute a generalized Ricci form.
In the following we will assume that $A$ is a positive constant.
The results are summarized as follows.
\begin{itembox}[l]{\textbf{Main result 2}}
  Let $g$ be written as \eqref{eq:RiemannianmetricG(y)=AE(y)}.
  Then generalised Ricci form is written as
\begin{equation*}
    \mathcal{R} = -\frac{1}{2} \pdv[2]{y} \qty(\log{E(y)}) dx \wedge dy.
  \end{equation*}
  $\nabla_Q$ is a quantum Levi-Civita connection for $g_1$ if $E(y)$ is written of the form :
\begin{alignat*}{2}
    E(y) & = c_2 \sec^2{\qty(\frac{\alpha (y + 2c_1)}{\sqrt{2}})},   & \quad  (\alpha > 0) \\
    E(y) & = \frac{c_2}{(y + 2c_1)^2},                               & \\
    E(y) & = c_2 \sech^2{\qty(\frac{-\alpha (y + 2c_1)}{\sqrt{2}})}, & \quad (\alpha < 0)
  \end{alignat*}
  where $c_1, c_2, \alpha$ are constants, in particular $c_2 > 0$.
\end{itembox}
From here in this section, we derive the main result 2.

\subsubsection{Generalised Ricci form}
If the Riemannian metric is given by \eqref{eq:RiemannianmetricG(y)=AE(y)}, 
we compute the generalized Ricci form with the same discussions as in section \ref{subsec:classicalupperhalfplane}.

The Christoffel symbols are
\begin{alignat*}{3}
  \Gamma^1_{11} & = 0,                     \quad & \Gamma^1_{12} & = \Gamma^1_{21} = \frac{E'(y)}{2E(y)}, \quad & \Gamma^1_{22} & = 0 \\
  \Gamma^2_{11} & = -\frac{E'(y)}{2AE(y)}, \quad & \Gamma^2_{12} & = \Gamma^2_{21} = 0,                   \quad & \Gamma^2_{22} & = \frac{E'(y)}{2E(y)}
\end{alignat*}
by substituting $G(y) = A E(y)$ in section \ref{subsec:classicalreduction}.
\newpage
\noindent
Using the Christoffel symbols just obtained, the Riemann curvature tensor is
\begin{align*}
  R^1_{221} & = (\Gamma^1_{12})' + \Gamma^1_{12} \Gamma^1_{21} - \Gamma^2_{22} \Gamma^1_{12} = -\frac{E'(y)^2 - E(y) E''(y)}{2E(y)^2}, \\
  R^1_{212} & = -R^1_{221}, \\
  R^2_{112} & = -(\Gamma^2_{11})' + \Gamma^1_{12} \Gamma^2_{11} - \Gamma^2_{11} \Gamma^2_{22} = -\frac{E'(y)^2 - E(y) E''(y)}{2AE(y)^2}, \\
  R^2_{121} & = -R^2_{112}.
\end{align*}

Next, we compute the components $\omega^{ij}$ of the inverse matrix of symplectic form.
From \eqref{eq:poissoncompatibility}, We only need to find $\omega^{12}$ satisfying
\begin{equation*}
  \left\{ \,
  \begin{aligned}
    & \pdv{\omega^{12}}{x} + \omega^{12} \Gamma^2_{21} + \omega^{12} \Gamma^1_{11} = 0 \\
    & \pdv{\omega^{12}}{y} + \omega^{12} \Gamma^1_{12} + \omega^{12} \Gamma^2_{22} = 0
  \end{aligned}
  \right.
\end{equation*}
and the solution is
\begin{equation*}
  \omega^{12} = \frac{c_1}{E(y)} (= -\omega^{21}). \quad (\text{$c_1$ is an arbitrary constant})
\end{equation*}
In the following we assume $c_1 = 1$.

To compute the generalized Ricci form we compute the 2-forms $H^{ij}$.
If the background connection is a Levi-Civita connection, then $H^{ij}$ is
\begin{align*}
  H^{11} & = -\frac{1}{2} \omega^{12} R^1_{212} dx \wedge dy \\
  H^{22} & = -\frac{1}{2} \omega^{21} R^2_{121} dy \wedge dx = -\frac{1}{2} \omega^{12} R^2_{121} dx \wedge dy
\end{align*}
from \eqref{eq:RiemannH^ij}.
From this, using the Riemann curvature tensor, the generalized Ricci form is
\begin{align}
  \mathcal{R} = H^{11} g_{11} + H^{22} g_{22} & = -\frac{1}{2} \omega^{12} (g_{11} R^1_{212} + g_{22} R^2_{121}) dx \wedge dy \nonumber \\
                                              & = -\frac{E'(y)^2 - E(y) E''(y)}{2 E(y)^2} dx \wedge dy \nonumber \\
                                              & = -\frac{1}{2} \pdv{y} \qty(\frac{E'(y)}{E(y)}) dx \wedge dy \nonumber \\
                                              & = -\frac{1}{2} \pdv[2]{y}(\log{E(y)}) dx \wedge dy.    \label{eq:R(G(y)=AE(y))}
\end{align}
From the above we obtain the first result of main result 2.

\subsubsection{Condition \texorpdfstring{$\nabla_Q$}{} is a quantum Levi-Civita connection}
From Proposition \ref{prop:quantumlevicivita}, the quantized connection $\nabla_Q$ is a quantum Levi-Civita connection and $\nabla \mathcal{R} = 0$ was equivalent.
Then we calculate $\nabla \mathcal{R}$ using \eqref{eq:R(G(y)=AE(y))}.
\newpage
\begin{align*}
  \nabla \mathcal{R} & = \nabla \qty{-\frac{1}{2} \pdv[2]{y} (\log{E(y)})} dx \wedge dy \\
                     & = -\frac{1}{2} \qty{ d \qty(\pdv[2]{y} (\log{E(y)})) dx \wedge dy + \qty(\pdv[2]{y} (\log{E(y)})) (\nabla dx \otimes dy - \nabla dy \otimes dx)} \\
                     & \quad -\frac{1}{2} \qty(\pdv[2]{y} (\log{E(y)})) (\sigma \otimes id) (dx \otimes \nabla dy - dy \otimes \nabla dx) \\
                     & = -\frac{1}{2} \pdv[3]{y} (\log{E(y)}) dy \otimes dx \wedge dy \\
                     & \quad -\frac{1}{2} \qty(\pdv[2]{y} (\log{E(y)})) (-\Gamma^1_{12} dx \otimes dy \otimes dy + \Gamma^2_{11} dx \otimes dx \otimes dx - \Gamma^1_{21} dy \otimes dx \wedge dy) \\
                     & \quad -\frac{1}{2} \qty(\pdv[2]{y} (\log{E(y)})) (-\Gamma^2_{11} dx \otimes dx \otimes dx + \Gamma^1_{12} dx \otimes dy \otimes dy - \Gamma^1_{21} dy \otimes dx \wedge dy) \\
                     & = -\frac{1}{4} \qty(\pdv[3]{y} (\log{E(y)}) - \pdv[2]{y} (\log{E(y)}) \pdv{y} (\log{E(y)})) dy \otimes dx \wedge dy \\
                     & = -\frac{1}{4} \pdv{y} (\pdv[2]{y} (\log{E(y)}) - \frac{1}{2} \qty(\pdv{y} (\log{E(y)}))^2) dy \otimes dx \wedge dy.
\end{align*}
From this we see that $\nabla \mathcal{R} = 0$, we need to satisfy
\begin{equation}
  \pdv[2]{y} (\log{E(y)}) - \frac{1}{2} \qty(\pdv{y} (\log{E(y)}))^2 = 2\alpha. \quad (\text{$\alpha$ is a constant})   \label{eq:ode}
\end{equation}
We take $Y = e^{-\frac{1}{2} \log{E(y)}} = E(y)^{-\frac{1}{2}}$ and 
multiplying $Y$ by both sides gives \eqref{eq:ode} is
\begin{align*}
  -2 \pdv[2]{y} Y         & = 2\alpha Y \\
  \pdv[2]{y} Y + \alpha Y & = 0.
\end{align*}
Thus, solutions are obtained when $\alpha$ is positive, zero, and negative, respectively, and returning to $E(y)$, we obtain
\begin{alignat*}{2}
  E(y) & = c_2 \sec^2{\qty(\frac{\alpha (y + 2c_1)}{\sqrt{2}})},   & \quad (\alpha > 0) \\
  E(y) & = \frac{c_2}{(y + 2c_1)^2},                               & \quad (\alpha = 0) \\
  E(y) & = c_2 \sech^2{\qty(\frac{-\alpha (y + 2c_1)}{\sqrt{2}})}, & \quad (\alpha < 0)
\end{alignat*}
where $c_1, c_2$ are an arbitrary constants, in particular $c_2 > 0$ since $E(y)$ is a positive function.
From the above we obtain the second result of main result 2.
We can see that the example of the upper half-plane corresponds to the case $\alpha = 0$, $c_1 = 0$, $c_2 = 1$.

When $E(y) = AG(y)$, there was a (positive) function degree of freedom of $E(y)$ in having the same form as classical one.
However, in order for the quantized connection $\nabla_Q$ to be a quantum Levi-Civita connection, 
we see that $E(y)$ must be one of the functions of the main result 2.

\subsection{Main result 3 : semiquantization with generalised Poincar\'{e} metric}    \label{subsec:G(y)=A(y)E(y)}
In section \ref{subsubsec:nabla_Q},
\begin{equation}
  G(y) = c_2 e^{2c_1 \sqrt{E(y)}} E(y)    \label{eq:gexpe}
\end{equation}
appeared as the condition that the $\lambda^1$ terms of $\nabla_Q dy$ are zero when $(m, n) = (2, 1)$.
Thus, in this section, 
we semiquantise the upper half-plane with
\begin{equation}
  g = y^{-2} (dx \otimes dx + e^{\frac{2t}{y}} dy \otimes dy)   \label{eq:riemannianmetric2}
\end{equation}
as the Riemannain metric when $E(y) = y^{-2}$, $c_1 = t(\in \mathbb{R})$, and $c_2 = 1$ in \eqref{eq:gexpe}.
Note that the deformation with parameter $t$ is reduced to the example in section 6.1 of \cite{name1} at $t \to 0$, 
and when $c_2 = c^2$, it is reduced to the example calculated at the beginning of this section at $t \to 0$.

The semiquantized results are summarized as follows.
\begin{itembox}[l]{\textbf{Main result 3}}
  The quantum Riemannian metric $g_1$ for the Riemannian metric \eqref{eq:riemannianmetric2} is
  \begin{align*}
    g_1 & = y^{-2} (dx \otimes_1 dx + e^{-\frac{2t}{y}} dy \otimes_1 dy) \\
        & \quad + \frac{\lambda}{2} y^{-3} \qty{ \{ (1 + 2e^{-\frac{t}{y}})t - 2e^{-\frac{t}{y}}y \} dx \otimes_1 dy + e^{-\frac{t}{y}} \{ (e^{-\frac{t}{y}} - 2)t + 2y \} dy \otimes_1 dx },
  \end{align*}
  $\nabla_Q dx$ and $\nabla_Q dy$ are given by
  \begin{align*}
    \nabla_Q dx & = y^{-1} (dx \otimes_1 dy + dy \otimes_1 dx) + \frac{\lambda}{2} e^{- \frac{t}{y}} y^{-2} t (e^{-\frac{2t}{y}} dx \otimes_1 dx + dy \otimes_1 dy), \\
    \nabla_Q dy & = - e^{-\frac{2t}{y}} y^{-1} dx \otimes_1 dx + y^{-2} (y + t)dy \otimes_1 dy + \frac{\lambda}{2} e^{-\frac{3t}{y}} t(t + y) y^{-3} dx \otimes_1 dy.
  \end{align*}

  The quantized connection $\nabla_Q$ is a weak quantum Levi-Civita connection for $g_1$ and, if $t \to 0$, it is a quantum Levi-Civita connection.
\end{itembox}
From here in this section, we derive the main result 3.

\subsubsection{Classical data}
First, the Christoffel symbols of the Levi-Civita connection for \eqref{eq:riemannianmetric2} are as follows :
\begin{align*}
  & \Gamma^1_{11} = \Gamma^1_{22} = 0, \quad \Gamma^1_{12} = \Gamma^1_{21} = -y^{-1}, \\
  & \Gamma^2_{12} = \Gamma^2_{21} = 0, \quad \Gamma^2_{11} = e^{-\frac{2t}{y}} y^{-1}, \quad \Gamma^2_{22} = -(t + y) y^{-2}.
\end{align*}

Next, we compute the components $\omega^{ij}$ of the inverse matrix of symplectic form.
From \eqref{eq:poissoncompatibility}, $\omega^{12}$ satisfying
\begin{equation*}
  \pdv{\omega^{12}}{y} - \omega^{12} \Gamma^1_{12} - \omega^{12} \Gamma^2_{22} = 0
\end{equation*}
is
\begin{equation*}
  \omega^{12} = a_1 e^{-\frac{t}{y}} y^2,
\end{equation*}
where $a_1$ is an arbitrary constant, but hereafter $a_1 = 1$.

Computing the Riemann curvature tensor from the Christoffel symbols, we obtain
\begin{equation*}
  R^1_{221} = (-t + y)y^{-3}, \quad R^2_{112} = e^{-\frac{2t}{y}}(-t + y) y^{-3}.
\end{equation*}
Scalar curvature is
\begin{equation*}
  R = g^{11} R^2_{121} + g^{22} R^1_{212} = 2e^{-\frac{2t}{y}}(t - y) y^{-1}.
\end{equation*}
The Christoffel symbols, $\omega^{ij}$, the Riemann curvature tensor, and the scalar curvature are reduced to the form calculated in section \ref{subsec:classicalupperhalfplane} at $t \to 0$.

Finally, we compute the generalized Ricci form.
By the same computation as in \eqref{eq:RiemannH^ij}, the 2-forms $H^{ij}$ are
\begin{align*}
  & H^{11} = -\frac{1}{2} e^{-\frac{t}{y}} y^2 (t - y) y^{-3} dx \wedge dy = -\frac{1}{2} e^{-\frac{t}{y}} (t - y) y^{-1} dx \wedge dy \\
  & H^{22} = \frac{1}{2} (-e^{-\frac{t}{y}} y^2) \{ -e^{-\frac{2t}{y}} (-t + y) y^{-3} \} dx \wedge dy = \frac{1}{2} e^{-\frac{3t}{y}} (-t + y) y^{-1} dx \wedge dy \\
  & H^{12} = H^{21} = 0
\end{align*}
so that the generalised Ricci form is
\begin{equation*}
  \mathcal{R} = g_{11} H^{11} + g_{22} H^{22} = e^{-\frac{t}{y}} (-t + y) y^{-3} dx \wedge dy.
\end{equation*}
The generalized Ricci form is also reduced to the form computed in section \ref{subsec:classicalupperhalfplane} at $t \to 0$.

\subsubsection{Quantum data}
First, we compute the $*$-product.
Using $\omega^{ij}$ computed in the previous section, the $*$-product is
\begin{equation*}
  f * g = fg + \frac{\lambda}{2} e^{-\frac{t}{y}} y^2 \qty(\pdv{f}{x} \pdv{g}{y} - \pdv{f}{y} \pdv{g}{x}).
\end{equation*}
Since the $*$-product is
\begin{alignat*}{3}
  x * x & = x^2,                                               & y * y & = y^2, \\
  x * y & = xy + \frac{\lambda}{2} e^{-\frac{t}{y}} y^2, \quad & y * x & = xy - \frac{\lambda}{2} e^{-\frac{t}{y}} y^2.
\end{alignat*}
for the coordinate functions $x$ and $y$, the commutation relation is
\begin{align*}
  [x, y]_* = \lambda e^{-\frac{t}{y}} y^2.
\end{align*}

Next, the $*$-product between a function and a 1-form is
\begin{align*}
  f * dx & = fdx + \frac{\lambda}{2} e^{-\frac{t}{y}} y \qty(\pdv{f}{x} dx - \pdv{f}{y} dy), \\
  f * dy & = fdy + \frac{\lambda}{2} e^{-\frac{t}{y}} y \qty(\pdv{f}{x} \left( t + y \right) y^{-1} dy + \pdv{f}{y} e^{-\frac{2t}{y}} dx).
\end{align*}
For the coordinate functions $x$ and $y$, the $*$-product is
\begin{alignat*}{2}
  & x * dx = xdx + \frac{\lambda}{2} e^{-\frac{t}{y}} ydx, & & \quad y * dx = ydx - \frac{\lambda}{2} e^{-\frac{t}{y}} ydy, \\
  & x * dy = xdy + \frac{\lambda}{2} e^{-\frac{t}{y}} (t + y) dy, & & \quad y * dy = ydy + \frac{\lambda}{2} e^{-\frac{3t}{y}}  ydx.
\end{alignat*}
\newpage
\noindent
The commutation relation between $x, y$ and $dx, dy$ is
\begin{align*}
  [x, dx]_* & = \lambda e^{-\frac{t}{y}} y dx = e^{\frac{2t}{y}} [y, dy]_*, \\
  [x, dy]_* & = \lambda e^{-\frac{t}{y}} (t + y) dy = -(t + y) y^{-1} [y, dx]_*.
\end{align*}

Third we compute the quantum wedge product.
By using $H^{ij}$, $dx \wedge_1 dx$, $dx \wedge_1 dy$, $dy \wedge_1 dx$, $dy \wedge_1 dy$ are
\begin{align*}
  dx \wedge_1 dx & = -\frac{\lambda}{2} e^{-\frac{t}{y}} (ty^{-1} + 1) dx \wedge dy, \\
  dx \wedge_1 dy & = dx \wedge dy, \\
  dy \wedge_1 dx & = dy \wedge dx, \\
  dy \wedge_1 dy & = -\frac{\lambda}{2} e^{-\frac{3t}{y}} (3ty^{-1} + 1) dx \wedge dy,
\end{align*}
respectively.

For the quantized metric $g_Q$, 
it is sufficient to compute the cases $(m, n) = (1, 2)$ and $(2, 1)$ from section \ref{subsubsec:g_Q}, and that is
\begin{subequations}
  \begin{align*}
    \intertext{$(m, n) = (1, 2)$ case}
    & \omega^{12} g_{11} (\Gamma^1_{12} \Gamma^2_{22} - \Gamma^1_{21} \Gamma^1_{12}) = ty^{-3} \\
    \intertext{$(m, n) = (2, 1)$ case}
    & \omega^{12} g_{22} (\Gamma^2_{11} \Gamma^1_{21} - \Gamma^2_{22} \Gamma^2_{11}) = e^{-\frac{2t}{y}} ty^{-3} \\
  \end{align*}
\end{subequations}
Therefore, the quantum Riemannian metric $g_1$ is
\begin{align*}
  g_1 & = g_Q - \lambda q^{-1}_{\Omega^1, \Omega^1} \mathcal{R} \\
      & = y^{-2} (dx \otimes_1 dx + e^{-\frac{2t}{y}} dy \otimes_1 dy) + \frac{\lambda}{2} ty^{-3} (dx \otimes_1 dy + e^{-\frac{2t}{y}} dy \otimes_1 dx) \\
      & \quad - \lambda e^{-\frac{t}{y}} (-t + y) y^{-3} (dx \otimes_1 dy - dy \otimes_1 dx) \\
      & = y^{-2} (dx \otimes_1 dx + e^{-\frac{2t}{y}} dy \otimes_1 dy) \\
      & \quad + \frac{\lambda}{2} y^{-3} \qty{ \{ (1 + 2e^{-\frac{t}{y}})t - 2e^{-\frac{t}{y}}y \} dx \otimes_1 dy + e^{-\frac{t}{y}} \{ (e^{-\frac{t}{y}} - 2)t + 2y \} dy \otimes_1 dx }.
\end{align*}

We then calculate $\nabla_Q dx^i$.
From the section \ref{subsubsec:nabla_Q}, 
we only need to compute the cases 
$(m, n) = (1, 1)$, $(m, n) = (2, 2)$ of $\nabla_Q dx$ and 
$(m, n) = (1, 2)$ of $\nabla_Q dy$.
For $\nabla_Q dx$
\begin{subequations}
  \begin{align*}
    \intertext{$(m, n) = (1, 1)$ case}
    & \frac{E'(y)^2 (-G(y) E'(y) + E(y) G'(y))}{8E(y)^2 G(y)^2} = -e^{-\frac{2t}{y}} t y^{-4} \\
    \intertext{$(m, n) = (2, 2)$ case}
    & \frac{E'(y)^2 (-G(y) E'(y) + E(y) G'(y))}{8E(y)^3 G(y)} = -t y^{-4}
  \end{align*}
\end{subequations}
For $\nabla_Q dy$
\begin{subequations}
  \begin{align*}
    \intertext{$(m, n) = (1, 2)$ case}
    & \frac{E'(y) G'(y) (-G(y) E'(y) + E(y) G'(y))}{8E(y) G(y)^3} = - e^{-\frac{2t}{y}} t(t + y) y^{-5} \\
  \end{align*}
\end{subequations}
From the above, $\nabla_Q dx^i$ is
\begin{align*}
  \nabla_Q dx & = y^{-1} (dx \otimes_1 dy + dy \otimes_1 dx) + \frac{\lambda}{2} e^{- \frac{t}{y}} y^{-2} t (e^{-\frac{2t}{y}} dx \otimes_1 dx + dy \otimes_1 dy), \\
  \nabla_Q dy & = - e^{-\frac{2t}{y}} y^{-1} dx \otimes_1 dx + y^{-2} (y + t)dy \otimes_1 dy + \frac{\lambda}{2} e^{-\frac{3t}{y}} t(t + y) y^{-3} dx \otimes_1 dy.
\end{align*}
To check whether $\nabla_Q$ satisfies $\nabla_Q g_1 = 0$, i.e. $\nabla_Q$ is a quantum Levi-Civita connection for $g_1$, we calculate $\nabla \mathcal{R}$.
We have
\begin{align}
  \nabla \mathcal{R} & = d (e^{-\frac{t}{y}} y^{-3}) dx \wedge dy + e^{-\frac{t}{y}} (-t + y) y^{-3} (\nabla dx \otimes dy - \nabla dy \otimes dx) \nonumber \\
                     & \quad + e^{-\frac{t}{y}} (-t + y) y^{-3} (\sigma \otimes id) (dx \otimes \nabla dy - dy \otimes \nabla dx) \nonumber \\
                     & = -y^{-5} e^{-\frac{t}{y}} (t^2 - 4ty + 2y^2) (dy \otimes dx \wedge dy) \nonumber \\
                     & \quad e^{-\frac{t}{y}} (-t + y) y^{-3} (y^{-1} dx \otimes dy \otimes dy + y^{-1} dy \otimes dx \otimes dy) \nonumber \\
                     & \quad -e^{-\frac{t}{y}} (-t + y) (-y^{-1} e^{-\frac{t}{y}} dx \otimes dx \otimes dx + (t + y) y^{-2} dy \otimes dy \otimes dx) \nonumber \\
                     & \quad e^{-\frac{t}{y}} (-t + y) y^{-3} (-y^{-1} e^{-\frac{2t}{y}} dx \otimes dx \otimes dx + (t + y) y^{-2} dy \otimes dx \otimes dy) \nonumber \\
                     & \quad -e^{-\frac{t}{y}} (-t + y) y^{-3} (y^{-1} dx \otimes dy \otimes dy + y^{-1} dy \otimes dy \otimes dx) \nonumber \\
                     & = \qty{ -y^{-5} e^{-\frac{t}{y}} (t^2 - 4ty + 2t^2) - e^{-\frac{t}{y}} (t - y) y^{-4} + e^{-\frac{t}{y}} (y^2 - t^2) y^{-5} } dy \otimes dx \wedge dy \nonumber \\
                     & = -y^{-5} t e^{-\frac{t}{y}} (2t - 3y) dy \otimes dx \wedge dy.   \label{eq:nablaR}
\end{align}
This means that $\nabla_Q$ is not a quantum Levi-Civita connection since $\nabla \mathcal{R} \neq 0$.
Then we check if $\nabla_Q$ is a weakly quantum Levi-Civita connection.
Calculating the cotorsion according to \eqref{eq:cotorsion_3}, we have
\begin{align*}
  \mathrm{co}T_{\nabla_Q} & = (\wedge \otimes id) (-\lambda y^{-5} t e^{-\frac{t}{y}} (2t - 3y) dy \otimes dx \wedge dy) \\
                 & = -\lambda y^{-5} t e^{-\frac{t}{y}} (2t - 3y) dy \wedge dx \wedge dy \\
                 & = 0.
\end{align*}
Therefore, $\nabla_Q$ is a weak quantum Levi-Civita connection.
Also, $\nabla \mathcal{R} \to 0$ when $t \to 0$ in \eqref{eq:nablaR}, 
and thus, $\nabla_Q$ is a quantum Levi-Civita connection to $g_1|_{t \to 0}$, 
where $g_1|_{t \to 0}$ is $g_1$ in the example of the upper half-plane of \cite{name1}.

\section{Conclusion and discussion}
In this paper we showed three results.
First, when the components of the Poincar\'{e} metric are arbitrary positive functions, 
the following conditions were obtained for the semiquantized results of 
$g_Q$ and $\nabla_Q dx^i$ to have the same form as classical one; 
for $g_Q$, $G(y) = A E(y)$($A$ is a positive constant), 
and for $\nabla_Q dx^i$, $G(y) = A E(y)$($A$ is a positive constant) or $G(y) = c_2 e^{2 c_1 \sqrt{E(y)}} E(y)$($c_1,\ c_2 (> 0)$ is constant).
This result shows that the upper half-plane is 
one of the examples that $g_Q$ and $\nabla_Q dx^i$ have the same form as classical ones.
In this study, we assumed that the function of the components depends only on $y$, 
but the case of a bivariate function, i.e. when it also depends on $x$, is a future work.

Second, if $G(y) = A E(y)$, then the generalized Ricci form $\mathcal{R}$ 
is written in the form of second derivatives of $\log{E(y)}$ and
there are three different forms of $E(y)$ that 
the quantized connection $\nabla_Q$ is a quantum Levi-Civita connection.
From this result, $E(y)$ can be any positive function in the main result 1, 
but for $\nabla_Q$ to be a quantum Levi-Civita connection, 
$E(y)$ is restricted to three types of functions.
When $E(y) = y^{-2}$, it is a Poincar\'{e} metric, 
but it remains to be seen what kind of surface it will be 
when $E(y) = c_2 \sec^2{\qty(\frac{\alpha (y + 2c_1)}{\sqrt{2}})}$ and $E(y) = c_2 \sech^2{\qty(\frac{-\alpha (y + 2c_1)}{\sqrt{2}})}$.

From main result 1 and main result 2, 
it was found that the upper half-plane was a special case in semiclassical theory 
because the semiquantization of $g_Q$ and $\nabla_Q$ has the same form as classical one and 
$\nabla_Q$ is a quantum Levi-Civita connection.

Third, we computed the example of the upper half-plane with a generalised Riemannian metric, 
which generalizes the Poincar\'{e} metric with the parameter $t$.
As a result, a terms which did not exist when $t = 0$ in the quantum Riemannian metric appeared, 
and $\nabla_Q$ became a weak quantum Levi-Civita connection.
And if $t \to 0$, it was reduced to the example of the upper half-plane in \cite{name1}.
This result shows that the deformation of the Riemannian metric with the parameter $t$ 
causes a deformation of the Christoffel symbols and Poisson structure, 
resulting in geometric degrees of freedom.

As a future prospect, we are considering its application to information geometry.
In this study, we were able to study the upper half-plane with the Riemannian metric including the Fisher metric of Gaussian distribution.
Furthermore, we were able to define the cotorsion and 
weak quantum Levi-Civita connection derived from affine geometry and 
construct the weak quantum Levi-Civita connection in the example of the upper half-plane.

In constructing a general theory of information geometry using semiclassical theory, 
we first need to describe various objects in information geometry with differential form only.
This is because information geometry is discussed using vector fields, 
but as mentioned above, in noncommutative geometry, 
it is difficult to define a vector field as one that satisfies the ``Leibniz rule''.
However, methods to treat vector fields in the context of noncommutative geometry have also been studied, 
mainly with two main strategies.
\begin{enumerate}
  \item Define the Leibniz rule for the left and right actions on a bimodule, respectively \cite{name5}.
  \item A vector field is considered as a velocity vector field on a geodesic over a bimodule \cite{name6}.
\end{enumerate}
The first is the Cartan pair, which is a generalized definition of the Leibniz rule satisfied by ordinary vector fields. \\
The second is to define a vector field as the element of a dual module of the differential form, 
and to characterize it as a velocity vector field by introducing geodesics on a bimodule.

After the objects of information geometry can be described only with differential form, 
we consider the following two strategies for constructing a general theory of information geometry using semiclassical theory.
\begin{enumerate}
  \item Semiquantize the classical Levi-Civita connection and define a semiquantum version of the $\alpha$-connection.
  \item Semiquantize statistical manifolds and develop information geometry on them.
\end{enumerate}
The first is to assume that the background connection satisfies $\nabla g = 0$ and 
to define the $\alpha$-connection in information geometry on a semiquantized space.
Since the definition of a dual connection on a bimodule is already known \cite{name7}, 
it seems possible to define a dual connection on a deformed bimodule using semiquantization.
Then, by defining a dual connection, we can consider the $\alpha$-connection.
Thus, we can construct a semiquantized version of the $\alpha$-connection on the semiquantized space. \\
The second is an attempt to set $\nabla g = C$ (where $C$ is a totally symmetric tensor field) as the background connection.
This is because there exist connections satisfying $\nabla g = C$ in statistical manifolds.
However, we think that semiquantization with a connection satisfying $\nabla g = C$ will cause difficulties.
This is because in semiclassical theory, $\nabla g = 0$ was required since the quantized metric is central 
(to make the results of contractions from the left and right coincide).
When we set $\nabla g = C$, we need to investigate how it affects the other parts of the results.

\subsection*{Acknowledgments}
I would like to express my appreciation to my supervisor Professor \mbox{Tatsuo} Suzuki.

\end{document}